\documentclass[a4paper,UKenglish,cleveref,autoref,thm-restate]{lipics-v2021}
\hideLIPIcs
\usepackage[T1]{fontenc}
\usepackage{graphicx}
\usepackage{amsmath} % \mathrel
\usepackage{amssymb} % \square
\usepackage{array} % for array env.
\usepackage{tikz-cd}
\usepackage{textgreek} % for textlambda
\usepackage{xspace} % for macrox.tex, but what?
\usepackage{stmaryrd} % for llbracket and rrbracket
\usepackage{mathpartir} % \inferrule
\usepackage{multicol} % \multicols
\usepackage{wrapfig} % \wrapfigure

\PassOptionsToPackage{colorlinks=true,allcolors=black}{hyperref} % color rather than box the text of links
\PassOptionsToPackage{capitalize,noabbrev,nameinlink}{cleveref} % load last, in particular after amsthm and hyperref, XXX: keep \crefname consistent with [no]capitalize option
\usepackage{glossaries-extra} % load after hyperref, \newacronym, \newglossaryentry, \glsfmtshort, \glsfmttext

\usepackage[bracketed=,macros]{lib/mymathtools}
\usepackage[bracketed=,macros]{lib/mycompounds}

\input{lib/macros}

\DeclareSymbolFont{sfoperators}{T1}{\sfdefault}{\mddefault}{\itdefault}

\makeatletter
\renewcommand{\operator@font}{\mathgroup\symsfoperators}
\makeatother

\renewcommand{\termfont}{\mathsf}
\title{Lax Modal Lambda Calculi}

\author{Nachiappan Valliappan}{University of Edinburgh, United Kingdom \and \url{http://nachivpn.me} }{nachivpn@gmail.com}{https://orcid.org/0000-0002-9358-3852}{}

\authorrunning{N. Valliappan}

\Copyright{Nachiappan Valliappan} %TODO mandatory, please use full first names. LIPIcs license is "CC-BY";  http://creativecommons.org/licenses/by/3.0/



\supplementdetails[subcategory={Formalization}]{Software}{https://github.com/nachivpn/s}

\funding{This work was funded by a Royal Society
  Newton International Fellowship.}

\acknowledgements{I thank Alex Kavvos, Andreas Abel, Carlos Tomé
  Cortiñas, Fabian Ruch and my colleagues at the University of
  Edinburgh for their feedback on this work. I also thank the
  anonymous reviewers of this paper and its previous versions for
  their valuable comments and suggestions.}% TODO

\nolinenumbers %uncomment to disable line numbering

\EventEditors{Stefano Guerrini and Barbara K\"{o}nig}
\EventNoEds{2}
\EventLongTitle{34th EACSL Annual Conference on Computer Science Logic (CSL 2026)}
\EventShortTitle{CSL 2026}
\EventAcronym{CSL}
\EventYear{2026}
\EventDate{February 23--28, 2026}
\EventLocation{Paris, France}
\EventLogo{}
\SeriesVolume{363}
\ArticleNo{49}
\begin{document}

\maketitle

\begin{abstract}
  Intuitionistic modal logics (IMLs) extend intuitionistic
  propositional logic with modalities such as the box and diamond
  connectives.
  Advances in the study of IMLs have inspired several applications in
  programming languages via the development of corresponding type
  theories with modalities.
  Until recently, IMLs with diamonds have been misunderstood as
  somewhat peculiar and unstable, causing the development of type
  theories with diamonds to lag behind type theories with boxes.
  In this article, we develop a family of typed\hyp{}lambda calculi
  corresponding to sublogics of a peculiar IML with diamonds known as
  Lax logic.
  These calculi provide a modal logical foundation for various strong
  functors in typed-functional programming.
  We present possible\hyp{}world and categorical semantics for these
  calculi and constructively prove normalization, equational
  completeness and proof\hyp{}theoretic inadmissibility results.
  Our main results have been formalized using the proof assistant Agda.
\end{abstract}

\section{Introduction}
\label{sec:intro}

In modal logic, a modality is a unary logical connective that exhibits
some logical properties.
Two such modalities are the connectives $\BoxTy{}$ (``box'') and
$\DiaTy{}$ (``diamond'').
Intuitively, a formula $\BoxTy{A}$ can be understood as ``necessarily
$A$'' and a formula $\DiaTy{A}$ as ``possibly $A$''.
In classical modal logic, the most basic logic~\CK extends classical
propositional logic (\CPL) with the box modality, the
\emph{necessitation} rule (if $A$ is a theorem then so is $\BoxTy{A}$)
and the \KA axiom ($\BoxTy{(A \FunTy B)} \FunTy \BoxTy{A} \FunTy
\BoxTy{B}$).
The diamond modality can be encoded in this logic as a dual of the box
modality: $\DiaTy{A} \equiv \neg{\BoxTy{\neg{A}}}$.
That is, $\DiaTy{A}$ is true if and only if $\neg{\BoxTy{\neg{A}}}$ is
true.
%

%
% $\BoxTy{A}$ is true at a world~$w$ if $A$
% is true at \emph{all} ($\forall$) possible ``future'' worlds, and
% $\DiaTy{A}$ is true a world~$w$ if $A$ is true at \emph{some}
% ($\exists$) possible future world.
%

In intuitionistic modal logic (IML), there is no consensus on one
logic as the most basic logic.
We instead find a variety of different IMLs based on different
motivations.
The $\BoxTy$ and $\DiaTy$ modalities are independent connectives in
IML \cite[Requirement 5]{Simpson94a}, just as $\land$ and $\lor$ are
independent connectives that are not inter\hyp{}definable in
intuitionistic propositional logic (\IPL).
In contrast to $\BoxTy{}$, however, the logical properties of
$\DiaTy{}$ vary widely in IML literature.
This has misconstrued $\DiaTy{}$ as a controversial and unstable
modality.
It had been incorrectly assumed until recently that
several IMLs with both $\BoxTy{}$ and $\DiaTy{}$ coincided (i.e. were
conservative extensions of their sublogics) only in the
$\DiaTy$\hyp{}free fragment, suggesting some sort of stability of
$\BoxTy{}$\hyp{}only logics.
Fortunately, misconceptions concerning intuitionistic diamonds have
been broken in recent results~\cite{DasM23,GrootSC25} and we are
approaching a better understanding of it.

Advances in IML have led to a plethora of useful applications in
programming languages through the development of corresponding type
theories with modalities.
Modal lambda calculi~\cite{PfenningD01,Clouston18} with box modalities
have found applications in staged
meta\hyp{}programming~\cite{DaviesP01,NanevskiPP08,HuP24}, reactive
programming~\cite{BahrGM19}, safe usage of temporal
resources~\cite{Ahman23} and checking productivity of recursive
definitions~\cite{BizjakGCMB16}.
Two particular box axioms that have received plenty of attention in
these developments are the axioms~$\TA : \BoxTy{A} \FunTy A$ and
$\FourA : \BoxTy{A} \FunTy \BoxTy{\BoxTy{A}}$.
Dual\hyp{}context modal calculi~\cite{PfenningD01,Kavvos17} which
admit one or both of these axioms are well\hyp{}understood.
These calculi enjoy a rich meta\hyp{}theory, including confluent
reduction, normalization and a comprehensive analysis of provability.
Fitch\hyp{}style modal lambda calculi~\cite{Clouston18} admitting
axioms~$\TA$ and $\FourA$ further enjoy an elegant categorical
interpretation, possible\hyp{}world semantics, and results showing how
categorical models of these calculi can be constructed using
possible\hyp{}world semantics of their corresponding
logics~\cite{ValliappanRC22}.

Lambda calculi with diamond modalities in comparison have received
much less attention from the type\hyp{}theoretic perspective.
The controversy surrounding the diamond modality in IML appears to
have restricted the development of type theories with diamonds.
For example, Kavvos~\cite{Kavvos20} cites Simpson's
survey~\cite{Simpson94a} of IMLs and restricts the development of
dual\hyp{}context modal calculi ``to the better-behaved, and seemingly
more applicable box modality'' arguing that the ``computational
interpretation [of $\DiaTy{}$] is not very crisp''.
Recent breakthroughs in intuitionistic modal logic have made it clear
that diamonds are no more problematic than boxes.
In this article, we further the type\hyp{}theoretic account of a
special class of diamond modalities with compelling applications in
programming languages.

\emph{Propositional lax logic} (\LL) is an intuitionistic modal logic
introduced independently by Fairtlough and
Mendler~\cite{FairtloughM97} and Benton, Bierman and de
Paiva~\cite{BentonBP98}.
\LL extends \IPL with a diamond modality~$\DiaTy{}$, known as the lax
modality, which exhibits a peculiar modal axiom~\SA (for
``strength''), in addition to axioms~\RA (for ``return'') and~\JA (for
``join'') that are well\hyp{}known in classical modal logic as duals
to the box axioms~$\TA$ and $\FourA$ respectively.
\begin{align*}
  \SA : A \ProdTy \DiaTy{B} \FunTy \DiaTy{(A \ProdTy B)} & &
  \RA : A \FunTy \DiaTy{A} & &
  \JA : \DiaTy{\DiaTy{A}} \FunTy \DiaTy{A}
\end{align*}
It is known that \LL corresponds to a typed\hyp{}lambda calculus (we
call \MLC) known as Moggi's \emph{monadic
  metalanguage}~\cite{Moggi91}, which models side effects in
functional programming using \emph{strong monads} from category
theory.
Benton, Bierman and de Paiva~\cite{BentonBP98}, and later Pfenning and
Davies~\cite{PfenningD01}, show that a judgment is provable in a
natural deduction proof system for \LL if and only if there exists a
typing derivation for its corresponding judgment in \MLC.
However, in contrast to the comprehensive treatment of box modalities
mentioned above, there remain several gaps in our understanding of the
lax modality:

\begin{enumerate}
\item It has remained unclear as to whether type theories can exist
  for sublogics of \LL or whether the axioms of \LL in combination
  happen to coincidentally enjoy a status of
  ``well\hyp{}behavedness''. What happens if we drop one or more of
  the modal axioms~$\RA$ and $\JA$? Does a corresponding type theory
  still exist?
  
\item A satisfactory account of the correspondence between the
  possible\hyp{}world semantics of \LL and the categorical semantics
  of \MLC is still missing. In particular, how can we leverage the
  possible\hyp{}world semantics of \LL to construct models of \MLC?
\end{enumerate}

The first objective of this article is to develop corresponding type
theories for sublogics of \LL that drop one or both of axioms~\RA and
\JA.
From the type\hyp{}theoretic perspective, this corresponds to type
theories for non\hyp{}monadic strong functors, which are prevalent in
functional programming.
For example, in Haskell, the array data type (in \texttt{Data.Array})
is a strong functor that neither exhibits return (axiom~\RA) nor join (axiom~\JA). %
Several other Haskell data types exhibit
return\footnote{\url{https://hackage.haskell.org/package/pointed-5.0.5/docs/Data-Pointed.html}}
or
join\footnote{\url{https://hackage.haskell.org/package/semigroupoids-6.0.1/docs/Data-Functor-Bind.html\#g:4}},
but not
both\footnote{\url{https://wiki.haskell.org/Why_not_Pointed\%3F}}.
We are interested in developing a uniform modal logical foundation for
the axioms of non-monadic strong functors.

The second objective of this article is to study the connection
between possible\hyp{}world semantics of \LL and its sublogics and
categorical models of their corresponding type theories.
Possible\hyp{}world semantics for logics are concerned with
provability of formulas and not about proofs themselves.
Categorical models of lambda calculi, on the other hand, distinguish
different proofs (terms) of the same proposition (type).
Mitchell and Moggi~\cite{MitchellM91} show the connection between
these two different semantics using a categorical refinement of
possible\hyp{}world semantics for the simply\hyp{}typed lambda
calculus (\STLC).
They note that their refined semantics, which we shall call
\emph{proof\hyp{}relevant possible\hyp{}world semantics}, makes it
``easy to devise Kripke counter\hyp{}models'' since they ``seem to
support a set\hyp{}like intuition about lambda terms better than
arbitrary cartesian closed categories''.
We wish to achieve this technical convenience in model construction
for all the modal lambda calculi in this article.

Towards our first objective, we formulate three new modal lambda
calculi as subsystems of \MLC: \SLC, \RLC, \JLC.
The calculus~\SLC models \emph{strong} functors and corresponds to a
logic~\SL (for ``S\hyp{}lax Logic'') that admits axiom~\SA, but
neither \RA nor \JA.
The calculus~\RLC models strong \emph{pointed} functors and
corresponds to a logic~\RL (for ``SR\hyp{}lax Logic'') that admits
axioms~\SA and~\RA, but not~\JA.
The calculus~\JLC models strong \emph{semimonads} and corresponds to a
logic~\JL (for ``SJ\hyp{}lax Logic'') that admits axioms~\SA and~\JA,
but not~\RA.
We refer to all four calculi collectively as \emph{lax modal lambda
  calculi}.
Towards our second objective, we extend Mitchell and Moggi's
proof\hyp{}relevant possible\hyp{}world semantics to lax modal lambda
calculi and show that it is complete for their equational theories.
We further show that all four calculi are normalizing by constructing
\emph{Normalization by Evaluation} models as instances of
possible\hyp{}world semantics and prove completeness and
inadmissibility results as corollaries.
All the theorems in this article have been verified
correct~\cite{Valliappan25} using the proof assistant
Agda~\cite{Agda2}.

\section{Overview  of \LL and its corresponding lambda calculus~\MLC}
\label{sec:overview}

In this section, we define the syntax and semantics of \LL and its
sublogics as extensions of the \emph{negative}, i.e. disjunction and
absurdity\hyp{}free, fragment of \IPL.
This section is a summary of the background presumed in this article
and is based on previously published
work~\cite{FairtloughM97,Moggi91}.
%
% The semantic correspondence between the \emph{positive} logical
% connectives disjunction and absurdity and their corresponding type
% formers in lambda calculi is a subtle matter in proof\hyp{}relevant
% possible\hyp{}world semantics, and is beyond the scope of this
% article.
%

\subsection{Syntax and semantics of \LL}
\emph{Syntax}. The language of (the negative fragment of) \LL consists of formulas
defined inductively by propositional atoms ($p$,~$q$,~$r$, etc.), a
constant~$\UnitTy$ and logical connectives~$\ProdTy$,~$\FunTy$
and~$\DiaTy$.
The connective~$\DiaTy$ has the highest operator precedence, and is
followed by $\ProdTy$ and $\FunTy$.
Following the usual convention, we suppose that $\ProdTy$ and~$\FunTy$
associate to the right.
\begin{align*}
     \Prop\qquad A,B ::= \ p,q,r,\ldots\ |\ \UnitTy\ |\ A \ProdTy
     B\ |\ A \FunTy B\ |\ \DiaTy{A} & &
         \Ctx\qquad \Gamma, \Delta ::= \EmptyCtx\ |\ \ExtCtx{\Gamma}{A}
\end{align*}
The constant~$\UnitTy$ denotes universal truth, the binary
connectives~$\ProdTy$ and~$\FunTy$ respectively denote conjunction and
implication, and the unary connective~$\DiaTy$ denotes the lax
modality.
Intuitively, a formula~$\DiaTy{A}$ may be understood as qualifying the
truth of formula~$A$ under \emph{some} constraint.
A context~$\Gamma$ is a multiset of formulas $A_1, A_2,..., A_n$,
where $\EmptyCtx$ denotes the empty context.

A Hilbert\hyp{}style axiomatization of \LL can be given by extending
the usual axioms and rules of deduction for \IPL with the modal
axioms~\SA,~\RA, and~\JA in \cref{sec:intro}.
%
%% \begin{multicols}{2}
%% \begin{enumerate}
%% \item $A \FunTy  B \FunTy  A$
%% \item $(A \FunTy  B \FunTy  C) \FunTy  (A \FunTy  B) \FunTy  A \FunTy  C$
%% \item $A \FunTy  B \FunTy  (A \ProdTy B)$
%% \item $(A \ProdTy B) \FunTy  A$
%% \item $(A \ProdTy B) \FunTy  B$
%% \item $A \ProdTy \DiaTy{B} \FunTy \DiaTy{(A \ProdTy B)}$
%% \item $A \FunTy \DiaTy{A}$
%% \item $\DiaTy{\DiaTy{A}} \FunTy \DiaTy{A}$
%% \end{enumerate}
%% %
%% \columnbreak
%% %
%% \begin{mathpar}

%% \inferrule[Axiom (Ax)]{
%%   A \in \mathcal{I}(Ax)
%% }{%
%%     \Gamma \vdash A
%% }\label{rule:ll-ax}%

%% \inferrule[Hypothesis (Hyp)]{%
%%   A \in \Gamma
%%   }{%
%%   \Gamma \vdash A
%% }\label{rule:ll-mp}%

%% \inferrule[Modus Ponens (MP)]{%
%%   \Gamma \vdash A \FunTy  B \\ \Gamma \vdash A
%%   }{%
%%     \Gamma \vdash B
%% }\label{rule:ll-mp}%

%% \end{mathpar}
%% \end{multicols}

\noindent\emph{Semantics}. The possible\hyp{}world semantics of \LL defines the truth of
\LL{}\nbhyp{}formulas in a model using gadgets known as \emph{frames}.
A \LL{}\nbhyp{}frame~$F = \tuple{W,\Ri,\Rm}$ is a triple that consists of
a set~$W$ of \emph{worlds} and two reflexive\hyp{}transitive
relations~$\Ri$ (for ``intuitionistic'') and~$\Rm$ (for ``modal'') on
worlds satisfying two compatibility conditions:
\begin{itemize}
  \item Forward confluence: ${\Ri^{-1} ; \Rm} \subseteq {\Rm ; \Ri^{-1}}$
  \item Inclusion: ${\Rm} \subseteq {\Ri}$
\end{itemize}
The relation~$\Ri^{-1}$ is the converse of relation~$\Ri$, and is
defined as ${\Ri^{-1}} = \{ (y,x)\ |\ (x,y) \in {\Ri} \}$.
The operator $;$ denotes composition of relations and is defined for
two relations~$R_1$ and~$R_2$ on worlds as
$R_1 ; R_2 = \{ (x,z)\ |\ \text{there exists}\ y \in W\ \text{such
  that}\ (x,y) \in R_1\ \text{and}\ (y,z) \in R_2 \}$.

We may intuitively understand worlds as nodes in a graph denoting the
``state of assumptions'', relation~$\Ri$ as paths denoting increase in
assumptions, and relation~$\Rm$ as paths denoting constraining of
assumptions.
That is, $w \Ri w'$ denotes the increase in assumptions from world~$w$
to $w'$, and $w \Rm v$ denotes a constraining of $w$ by $v$ such that
$v$ is reachable from $w$ when the constraint can be satisfied.
Under this reading, the inclusion condition~${\Rm} \subseteq {\Ri}$
states that imposing a constraint increases assumptions.

\begin{wrapfigure}{r}{0.3\textwidth}
\begin{center}
  \begin{tikzcd}[scale=2, sep=huge]
    w \arrow[r, "\Rm"] \arrow[d, "\Ri",swap] & v \arrow[d, "\Ri",dashrightarrow]\\
    w' \arrow[r, "\Rm",dashrightarrow] & v'
  \end{tikzcd}
\end{center}
%\caption{Forward confluence}
\label{fig:fc}
\end{wrapfigure}
The forward confluence condition~${\Ri^{-1} ; \Rm} \subseteq {\Rm ;
  \Ri^{-1}}$ states that constraints can be ``transported'' over an
increase in assumptions.
It can be visualized as depicted on the right, where the dotted lines
represent ``there exists''.
This condition does not appear in Fairtlough and Mendler's original
work~\cite{FairtloughM97}, but can be found in earlier work on
intuitionistic diamonds by Bo{\v z}i{\'c} and Do{\v s}en~\cite[\S
  8]{BozicD84} and Plotkin and Stirling~\cite{PlotkinS86}.
It simplifies the interpretation of $\DiaTy{}$ and is satisfied by all
the models we will construct in this article to prove completeness.
We return to the discussion on forward confluence in
\cref{sec:related}.

A model~$\Mod[M] = \tuple{F,V}$ couples a frame~$F$ with a
\emph{valuation} function~$V$ that assigns to each propositional
atom~$p$ a set~$V(p)$ of worlds hereditary in $\Ri$, i.e. if $w \Ri
w'$ and $ w \in V(p)$ then $w' \in V(p)$.
The truth of a formula in a model~$\Mod[M]$ is defined by the
\emph{satisfaction} relation~$\Vdash$ for a given world~$w \in W$ by
induction on a formula as:
\begin{equation*}
  \begin{array}{l@{\;\Vdash\;} @{\;} l @{\;\text{iff}\;}c@{\;} l}
    \Mod[M],w & p & & w \in V(p) \\
    \Mod[M],w & \UnitTy & & \text{true} \\
    \Mod[M],w & A \ProdTy B & & \Mod[M],w \Vdash A\ \text{and}\  \Mod[M],w \Vdash B \\
    \Mod[M],w & A \FunTy B & & \text{for all}\ w' \in W\ \text{such that}\ w \Ri w',\ \Mod[M],w' \Vdash A\ \text{implies}\ \Mod[M],w' \Vdash B\ \\
    \Mod[M],w & \DiaTy{A} & & \text{there exists}\ v \in W\ \text{with}\ w \Rm v\  \text{and}\ \Mod[M],v \Vdash A
  \end{array}
\end{equation*}
We write $\Mod[M],w \Vdash \Gamma$ to denote $\Mod[M], w \Vdash A_i$
for all formulas~$A_i$ with $1 \leq i \leq n $ in context~$\Gamma =
A_1, ...A_n$, and write $\Gamma \satisfies A$ to denote $\Mod[M],w
\Vdash \Gamma$ implies $\Mod[M],w \Vdash A$ for all worlds~$w$ in all
models~$\Mod[M]$.
Furthermore, we write $\Mod[M] \satisfies A$ to denote $\Mod[M],w
\Vdash A$ for all worlds~$w$ in $\Mod[M]$.

The soundness of \LL for its semantics can be shown using the
following key properties:
\begin{proposition}\label{prop:axioms}
  For an arbitrary model~$\Mod[M]$ of \LL
  \begin{itemize}
\item if $w \Ri w'$
  and $\Mod[M],w \Vdash A$ then $\Mod[M],w' \Vdash A$, for all worlds~$w,w'$ and formulas~$A$
\item $\Mod[M] \satisfies A \ProdTy \DiaTy{B} \FunTy \DiaTy{(A \ProdTy
  B)}$, for all formulas~$A,B$ %since ${\Rm} \subseteq {\Ri}$
\item $\Mod[M] \satisfies A \FunTy \DiaTy{A}$, for all formulas~$A$ %since $\Rm$ is reflexive
\item $\Mod[M] \satisfies \DiaTy{\DiaTy{A}} \FunTy \DiaTy{A}$, for all formulas~$A$ %since $\Rm$ is transitive
\end{itemize}
\end{proposition}

\begin{proof}
  The first property, sometimes called ``monotonicity'', states that
  the truth of a formula~$A$ persists as knowledge increases. This
  property can be proved by induction on the formula~$A$, using the
  forward confluence condition for the case of $\DiaTy{A}$. The
  remaining properties can be proved using the definition of the
  relations~$\satisfies$ and $\Vdash$, by respectively using the
  inclusion condition~${\Rm} \subseteq {\Ri}$, reflexivity of $\Rm$,
  and transitivity of $\Rm$.
\end{proof}

\subsection{Syntax and semantics of \MLC}
\label{sec:overview:mlc}

\emph{Syntax}. The calculus \MLC is a typed $\lambda$\hyp{}calculus
that was developed by Moggi~\cite{Moggi91} before \LL.
The language of \MLC consists of types, contexts and terms, and can be
understood as an extension of \STLC with a unary type
constructor~$\DiaTy$ that exhibits the \LL axioms~\SA,~\RA and~\JA.
Types and contexts in \MLC are defined inductively by the following grammars:
\begin{align*}
  \Ty\quad A,B ::= \BaseTy\ |\ \UnitTy\ |\ A \ProdTy B\ |\ A \FunTy B\ |\ \DiaTy{A} \qquad
  & \Ctx\quad \Gamma, \Delta ::= \EmptyCtx\ |\ \ExtCtx{\Gamma}{x : A} \quad (x\ \text{not in}\ \Gamma)
\end{align*}
The type~$\BaseTy$ denotes an uninterpeted base (or ``ground'') type,
$\UnitTy$ denotes the unit type, $A \ProdTy B$ denotes product types,
$A \FunTy B$ denotes function types, and $\DiaTy{A}$ denotes
\emph{modal} types.
A modal type~$\DiaTy{A}$ can be understood as the type of a
computation that performs some side\hyp{}effects to return a value of
type~$A$.
A context (or ``typing environment'')~$\Gamma$ is a list~$x_1 :
A_1,..., x_n : A_n$ of unique type assignments to variables, where
$\EmptyCtx$ denotes the empty context.

The term language of \STLC consists of variables ($x, y,...$) and term
constructs for the unit type ($\unitTm$), product types ($\pairTm$,
$\fstTm$, $\sndTm$) and function types ($\lamLabel$, $\appTm$).
The term language of \MLC extends that of \STLC with the
constructs~$\returnLLTm$ and~$\letLLTm$ for the modal types.
%
% \begin{align*}
%    \RawTm\ t,u ::=  \ x, y, \ldots\ |\ \unitTm\ |\ \pairTm{t,u}\ |\ \fstTm{t}\
%        |\ \sndTm{t}\ |\ \lamTm{x}{t}\ |\ \appTm{t,u}\ |\ \returnLLTm{x}\ |\ \letInLLTm{x}{t}{u}
% \end{align*}
%
The \emph{typing judgments}~$\Gamma \vdash t : A$ defined in
\cref{fig:calculus-mlc} identify \emph{well\hyp{}typed} terms in \MLC.
We say that a term~$t$ is \emph{well\hyp{}typed for} type~$A$
\emph{under} context~$\Gamma$ when there exists a derivation, called
the \emph{typing derivation}, of the typing
judgment~$\Gamma \vdash t : A$.
In this article, we are only concerned with well\hyp{}typed terms and
assume that every term~$t$ is well\hyp{}typed for some type~$A$ under
some context~$\Gamma$.
The \emph{equality judgments}~$\Gamma \vdash t \thyeq t' : A$ defined
in \cref{fig:calculus-mlc} identify the equivalence between
well\hyp{}typed terms~$t$ and $t'$ and specify the \emph{equational
  theory} for \MLC.

In \cref{fig:calculus-mlc}, the notation~$\subst{t}{u/x}$ denotes the
\emph{substitution} of term~$u$ for the variable~$x$ in term~$t$, and
the notation~$\wkTm{t}$ denotes the \emph{weakening} of a term~$\Gamma
\vdash t : A$ by embedding it into a larger context~$\Gamma \subseteq
\Gamma'$ as~$\Gamma' \vdash \wkTm{t} : A$.
We write $\Gamma \subseteq \Gamma'$ when the context~$\Gamma$ is a
sub\hyp{}list of context~$\Gamma'$, meaning $\Gamma'$ contains at
least the variable\hyp{}type assignments in $\Gamma$.
Substitution and weakening are both well\hyp{}typed operations on
terms that are admissible in \MLC:

\begin{itemize}
\item Substitution: If $\ExtCtx{\Gamma}{x : A} \vdash t : B$ and $\Gamma \vdash u : A$, then $\Gamma \vdash \subst{t}{u/x} : B$
\item Weakening: If $\Gamma \vdash t : A$ and $\Gamma \subseteq \Gamma'$, then $\Gamma' \vdash \wkTm{t} : A$
\end{itemize}

We say that a type~$X$ is \emph{derivable} (or \emph{can be derived})
in \MLC when there exists a typing derivation of the judgment
$\EmptyCtx \vdash t : X$ for some term~$t$.
The types corresponding to the \LL axioms~$\SA$, $\RA$ and $\JA$ are
each derivable in \MLC, as witnessed by the terms below:

\begin{itemize}
\item $\EmptyCtx \vdash \lamTm{x}{\letInLLTm{y}{(\sndTm{x})}{(\returnLLTm{(\pairTm{(\fstTm{x}),y}}))}} : A \ProdTy \DiaTy{B} \FunTy \DiaTy{(A \ProdTy B)}$
\item $\EmptyCtx \vdash \lamTm{x}{\returnLLTm{x}} : A \FunTy \DiaTy{A} $
\item $\EmptyCtx \vdash \lamTm{x}{\letInLLTm{y}{x}{y}} : \DiaTy{\DiaTy{A}} \FunTy \DiaTy{A}$
\end{itemize}

\begin{figure}[t]
  \begin{mathpar}
    \inferrule[Var\nbhyp{}Zero]{%
    }{%
      \ExtCtx{\Gamma}{x : A} \vdashVar x : A
    }%

    \inferrule[Var\nbhyp{}Succ]{%
      \Gamma \vdashVar x : A \\ (y\ \text{not in}\ \Gamma) 
    }{%
      \ExtCtx{\Gamma}{y : B} \vdashVar x : A
    }%

    \inferrule[Var]{%
      \Gamma \vdashVar x : A
    }{%
      \Gamma \vdash x : A
    }
    
    \inferrule[$\UnitTy$\nbhyp{}Intro]{%
    }{%
      \Gamma \vdash \unitTm : \UnitTy
    }%

    \inferrule[$\ProdTy$\nbhyp{}Intro]{%
      \Gamma \vdash t : A \\
      \Gamma \vdash u : B
    }{%
      \Gamma \vdash \pairTm{t,u} : A \ProdTy B
    }%

    \inferrule[$\ProdTy$\nbhyp{}Elim\nbhyp{}1]{%
      \Gamma \vdash t : A \ProdTy B
    }{%
      \Gamma \vdash \fstTm{t} : A
    }%

    \inferrule[$\ProdTy$\nbhyp{}Elim\nbhyp{}2]{%
      \Gamma \vdash t : A \ProdTy B
    }{%
      \Gamma \vdash \sndTm{t} : B
    }\\

    \inferrule[$\FunTy$\nbhyp{}Intro]{
      \ExtCtx{\Gamma}{x : A} \vdash t : B
    }{%
      \Gamma \vdash \lamTm{x}{t} : A \FunTy B
    }%

    \inferrule[$\FunTy$\nbhyp{}Elim]{%
      \Gamma \vdash t : A \FunTy B\\
      \Gamma \vdash u : A
    }{%
      \Gamma \vdash \appTm{t,u} : B
    }%

    \inferrule[\ML/$\DiaTy$\nbhyp{}Return]{%
      \Gamma  \vdash t : A
    }{%
      \Gamma \vdash \returnLLTm{t} : \DiaTy{A}
    }%

    \inferrule[\ML/$\DiaTy$\nbhyp{}Let]{%
      \Gamma  \vdash t : \DiaTy{A} \\
      \ExtCtx{\Gamma}{x : A} \vdash u : \DiaTy{B}
    }{%
      \Gamma \vdash \letInLLTm{x}{t}{u} : \DiaTy{B}
    }%
    
    \inferrule[$\UnitTy$\nbhyp{}$\eta$]{ %
      \Gamma \vdash t : \UnitTy
    }{%
      \Gamma \vdash t \thyeq \unitTm : \UnitTy
    }%

    \inferrule[$\ProdTy$\nbhyp{}$\eta$]{ %
      \Gamma \vdash t : A \ProdTy B\\
    }{%
    \Gamma \vdash t \thyeq \pairTm{(\fstTm{t}),(\sndTm{t})} : A \ProdTy B
    }%

    \inferrule[$\ProdTy$\nbhyp{}$\beta_1$]{ %
      \Gamma \vdash t : A\\
      \Gamma \vdash u : B\\
    }{%
    \Gamma \vdash \fstTm{(\pairTm{t,u})} \thyeq t : A
    }%

    \inferrule[$\ProdTy$\nbhyp{}$\beta_2$]{ %
      \Gamma \vdash t : A\\
      \Gamma \vdash u : B\\
    }{%
    \Gamma \vdash \sndTm{(\pairTm{t,u})} \thyeq u : B
    }%

    \inferrule[$\FunTy$\nbhyp{}$\eta$]{ %
      \Gamma \vdash t : A \FunTy B\\
    }{%
    \Gamma \vdash t \thyeq \lamTm{x}{(\appTm{(\wkTm{t}),x})} : A \FunTy B
    }%

    \inferrule[$\FunTy$\nbhyp{}$\beta$]{ %
      \ExtCtx{\Gamma}{x : A} \vdash t : B\\
      \Gamma \vdash u : A
    }{%
      \Gamma \vdash \appTm{(\lamTm{x}{t}),u} \thyeq \subst{t}{u/x} : B
    }%

    \inferrule[\ML/$\DiaTy$\nbhyp{}$\beta$]{ %
      \Gamma \vdash t : A \\
      \ExtCtx{\Gamma}{x : A} \vdash u : \DiaTy{B}
    }{%
      \Gamma \vdash \letInLLTm{x}{(\returnLLTm{t})}{u} \thyeq \subst{u}{t/x} : \DiaTy{B}
    }\label{rule:dia-beta/LL}%

    \inferrule[\ML/$\DiaTy$\nbhyp{}$\eta$]{ %
      \Gamma \vdash t : \DiaTy{A}
    }{%
      \Gamma \vdash t \thyeq \letInLLTm{x}{t}{(\returnLLTm{x})} : \DiaTy{A}
    }\label{rule:dia-eta/LL}%

    \inferrule[\ML/$\DiaTy$\nbhyp{}ass]{ %
      \Gamma \vdash t : \DiaTy{A} \\
      \ExtCtx{\Gamma}{x : A} \vdash u : \DiaTy{B} \\
      \ExtCtx{\Gamma}{y : B} \vdash u' : \DiaTy{C}
    }{%
      \Gamma \vdash \letInLLTm{y}{(\letInLLTm{x}{t}{u})}{u'} \thyeq \letInLLTm{x}{t}{(\letInLLTm{y}{u}{(\wkTm{u'})})} : \DiaTy{C}
    }\label{rule:dia-ass/LL}%
    
  \end{mathpar}
  \caption{Well\hyp{}typed terms and equational theory for \MLC}
  \label{fig:calculus-mlc}
\end{figure}

\noindent\emph{Semantics}. The semantics of \MLC is given using categories.
A categorical model of \MLC is a cartesian\hyp{}closed category
equipped with a strong monad~$\DiaFun{}$ (defined in \cref{app:def}).
Given a categorical model~$\Cat[C]$ of \MLC, we interpret types and
contexts in \MLC as $\Cat[C]$\hyp{}objects and terms~$\Gamma \vdash t
: A$ in \MLC as $\Cat[C]$\hyp{}morphisms~$\eval{t} : \eval{\Gamma}
\rightarrow \eval{A}$, by induction on types, contexts and terms
respectively.
The interpretation of the term constructs~$\returnLLTm{}$
and~$\letLLTm{}$ (and in turn the modal axioms~$\SA$, $\RA$ and $\JA$)
is given by the structure of the strong monad~$\DiaFun{}$.
%
% We refer the reader to the accompanying Agda mechanization for
% further details.

\begin{proposition}[Categorical semantics for \MLC]\label{prop:cs:mlc}
  Given two terms~$t,u$ in \MLC, $\Gamma \vdash t \thyeq u : A$ if and
  only if for all categorical models~$\Cat[C]$ of \MLC $\eval{t} =
  \eval{u} : \eval{\Gamma} \rightarrow \eval{A} $ in $\Cat[C]$.
\end{proposition}
\begin{proof}
  Follows by induction on the judgment~$\Gamma \vdash t \thyeq u : A$
  in one direction, and by a term model construction (see for e.g.,
  \cite[Section 3.2]{Clouston18}) in the converse.
\end{proof}
\subsection{Sublogics of \LL and corresponding lambda calculi}
\label{sec:overview:sublogics}

The minimal sublogic of \LL, which we call \SL, can be axiomatized by
extending the usual axioms and rules of \IPL with (only) the modal
axiom~\SA.
Furthermore, we axiomatize:
\begin{itemize}
\item the logic \RL by extending \SL with axiom~$\RA$
\item the logic \JL by extending \SL with axiom~$\JA$
\item the logic \LL by extending \SL with axioms~$\RA$ and $\JA$ (as
  defined previously)
\end{itemize}
The semantics for \SL, \RL and \JL is given as before for \LL by
restricting the definitions of frames.
An \SL{}\nbhyp{}frame~$F = \tuple{W,\Ri,\Rm}$ is a triple that
consists of a set~$W$ of worlds, a reflexive\hyp{}transitive
relation~$\Ri$, and a relation~$\Rm$ (that need not be reflexive or
transitive), satisfying the forward confluence and inclusion
conditions.
Furthermore, an \SL{}\nbhyp{}frame is
\begin{itemize}
\item an \RL{}\nbhyp{}frame when $\Rm$ is reflexive
\item an \JL{}\nbhyp{}frame when $\Rm$ is transitive
\item a \LL{}\nbhyp{}frame when $\Rm$ is reflexive and transitive (as defined previously)
\end{itemize}

%% The soundness and completeness of \LL (\cref{prop:ll-sound-complete})
%% can also be extended to \SL, \RL and \JL for models determined
%% respectively by \SL{}\hyp{}frames, \RL{}\hyp{}frames and
%% \JL{}\hyp{}frames.
%% %
%% We will instead prove them for their corresponding lambda calculi,
%% which subsumes this proof.
%% %

In the upcoming section, we will define a corresponding modal lambda
calculus for each of \LL's sublogics (\cref{sec:subcalculi}).
We develop proof\hyp{}relevant possible\hyp{}world semantics for these
calculi and show the connection to categorical semantics by studying
the properties of presheaf categories determined by
proof\hyp{}relevant frames (\cref{sec:prks}).
We leverage this connection to then construct Normalization by
Evaluation models for the calculi, and show as corollaries
completeness and inadmissibility theorems (\cref{sec:nci}).

\section{The calculi \SLC, \RLC and \JLC}
\label{sec:subcalculi}

We define the calculi~\SLC, \RLC and \JLC---akin to the calculus~\MLC
in \cref{sec:overview:mlc}---as extensions of \STLC with a unary type
constructor~$\DiaTy$ that exhibits the characteristic axioms of their
corresponding logics.
The types and contexts of all four calculi are identical to \MLC,
while the term constructs and equational theory for the respective
modal fragments vary.
\linebreak\linebreak
\noindent\emph{The calculus~\SLC}.
The well\hyp{}typed terms and equational theory of the modal fragment
of \SLC are defined in \cref{fig:calculus-slc}.
The calculus \SLC extends \STLC with a construct~$\letSLTm$ and two
equations~\SL/$\DiaTy$\nbhyp{}$\eta$ and~\SL/$\DiaTy$\nbhyp{}$\beta$.
Observe that the typing rule for $\letSLTm$ in \SLC differs from
$\letLLTm$ in \MLC: a term~$\letInSLTm{x}{t}{u}$ ``maps'' a
term~$\ExtCtx{\Gamma}{x : A} \vdash u : B$ over a
term~$\Gamma \vdash t : \DiaTy{A}$ to yield a term well\hyp{}typed for
type~$\DiaTy{B}$ under context~$\Gamma$.
This difference disallows a typing derivation for axiom~$\JA$, while
allowing axiom~$\SA$ to be derived as below:
$$\EmptyCtx \vdash \lamTm{x}{\letInSLTm{y}{(\sndTm{x})}{(\pairTm{(\fstTm{x}),y})}} : A \ProdTy \DiaTy{B} \FunTy
\DiaTy{(A \ProdTy B)}$$
Note how this derivation differs from the one in
\cref{sec:overview:mlc} for \MLC: it uses $\letSLTm$ in place of
$\letLLTm$ without a need for $\returnLLTm$.
Axiom~$\RA$, however, cannot be derived in \SLC.

A categorical model of \SLC is a cartesian\hyp{}closed category
equipped with a strong functor~$\DiaTy{}$ (that need not be a monad).
Given a categorical model~$\Cat[C]$ of \SLC, we interpret types and
contexts in \SLC as $\Cat[C]$\hyp{}objects and
terms~$\Gamma \vdash t : A$ as
$\Cat[C]$\hyp{}morphisms~$\eval{t} : \eval{\Gamma} \rightarrow
\eval{A}$ as before with \MLC by induction on types and terms
respectively.
The interpretation of the term construct~$\letSLTm{}$ (and in turn the
modal axiom~$\SA$) is given by the tensorial strength of
functor~$\DiaTy{}$, which gives us a morphism~$X \ProdTy \DiaTy{Y}
\rightarrow \DiaTy{(X \ProdTy Y)}$ for all objects~$X,Y$ in $\Cat[C]$.

\begin{proposition}[Categorical semantics for \SLC]\label{prop:cs:slc}
  Given two terms~$t,u$ in \MLC, $\Gamma \vdash t \thyeq u : A$ if and
  only if for all categorical models~$\Cat[C]$ of \SLC $\eval{t} =
  \eval{u} : \eval{\Gamma} \rightarrow \eval{A} $ in $\Cat[C]$.
\end{proposition}

\begin{figure}[t]
  \begin{mathpar}
    \inferrule[\SL/$\DiaTy$\nbhyp{}Letmap]{%
      \Gamma  \vdash t : \DiaTy{A} \\
      \ExtCtx{\Gamma}{x : A} \vdash u : B
    }{%
      \Gamma \vdash \letInSLTm{x}{t}{u} : \DiaTy{B}
    }%

    \inferrule[\SL/$\DiaTy$\nbhyp{}$\eta$]{ %
      \Gamma \vdash t : \DiaTy{A} }{%
      \Gamma \vdash t \thyeq \letInSLTm{x}{t}{x} : \DiaTy{A}
    }%

    \inferrule[\SL/$\DiaTy$\nbhyp{}$\beta$]{ %
      \Gamma \vdash t : \DiaTy{A} \\
      \ExtCtx{\Gamma}{x : A} \vdash u : B \\
      \ExtCtx{\Gamma}{y : B} \vdash u' : C
    }{%
      \Gamma \vdash \letInSLTm{y}{(\letInSLTm{x}{t}{u})}{u'} \thyeq \letInSLTm{x}{t}{\subst{(\wkTm{u'})}{u/y}} : \DiaTy{C}
    }%
  \end{mathpar}
  \caption{Well\hyp{}typed terms and equational theory for \SLC (omitting those of \STLC)}
  \label{fig:calculus-slc}
\end{figure}

\noindent\emph{The calculus~\RLC}.
The well\hyp{}typed terms and equational theory for the modal fragment
of \RLC are defined in \cref{fig:calculus-rlc}.
The calculus \RLC extends \STLC with two constructs~$\returnRLTm{}$
and~$\letRLTm$, and three equations~\RL/$\DiaTy$\nbhyp{}$\eta$,
\RL/$\DiaTy$\nbhyp{}$\beta_1$ and \RL/$\DiaTy$\nbhyp{}$\beta_2$.

Observe that the typing rule of the construct~$\letRLTm$ in \RLC is
identical to $\letSLTm$ in \SLC.
As a result, axiom~$\SA$ can be derived in \RLC similar to \SLC:
$$\EmptyCtx \vdash \lamTm{x}{\letInRLTm{y}{(\sndTm{x})}{(\pairTm{(\fstTm{x}),y})}} : A \ProdTy \DiaTy{B} \FunTy
\DiaTy{(A \ProdTy B)}$$
Similarly, observe that the typing rule of the construct~$\returnRLTm$
in \RLC is identical to $\returnLLTm$ in \MLC.
As a result axiom $\RA$ can as well be derived in \RLC similar to
\MLC:
$$\EmptyCtx \vdash \lamTm{x}{\returnRLTm{x}} : A \FunTy \DiaTy{A}$$
Axiom~$\JA$ cannot be derived in \RLC since there is no counterpart
for $\letLLTm$ in \RLC.

A categorical model of \RLC is a cartesian\hyp{}closed category
equipped with a strong pointed functor~$\DiaTy{}$.
The term construct~$\letRLTm{}$ (and in turn axiom~$\SA$) is
interpreted in a model~$\Cat[C]$ of \RLC using the tensorial strength
of functor~$\DiaTy{}$, as before with \SLC.
The interpretation of the term construct~$\returnRLTm{}$ (and in turn
axiom~$\RA$) is given by the pointed structure of the
functor~$\DiaTy{}$, which gives us a morphism~$X \rightarrow
\DiaTy{X}$ for all objects~$X$ in $\Mod[C]$.

\begin{proposition}[Categorical semantics for \RLC]\label{prop:cs:rlc}
  Given two terms~$t,u$ in \RLC, $\Gamma \vdash t \thyeq u : A$ if and
  only if for all categorical models~$\Cat[C]$ of \RLC $\eval{t} =
  \eval{u} : \eval{\Gamma} \rightarrow \eval{A} $ in $\Cat[C]$.
\end{proposition}

\begin{figure}[t]
  \begin{mathpar}
    \inferrule[\RL/$\DiaTy$\nbhyp{}Return]{%
      \Gamma  \vdash t : A
    }{%
      \Gamma \vdash \returnRLTm{t} : \DiaTy{A}
    }\label{rule:returnRLTm/RL}

    \inferrule[\RL/$\DiaTy$\nbhyp{}Letmap]{%
      \Gamma  \vdash t : \DiaTy{A} \\ \ExtCtx{\Gamma}{x : A} \vdash u : B
    }{%
      \Gamma \vdash \letInRLTm{x}{t}{u} : \DiaTy{B}
    }\label{rule:letTm/RL}

    \inferrule[\RL/$\DiaTy$\nbhyp{}$\eta$]{ %
      \Gamma \vdash t : \DiaTy{A}
    }{%
      \Gamma \vdash t \thyeq \letInRLTm{x}{t}{x} : \DiaTy{A}
    }\label{rule:dia-eta/RL}%

    \inferrule[\RL/$\DiaTy$\nbhyp{}$\beta_1$]{ %
      \Gamma \vdash t : \DiaTy{A} \\
      \ExtCtx{\Gamma}{x : A} \vdash u : B \\
      \ExtCtx{\Gamma}{y : B} \vdash u' : C
    }{%
      \Gamma \vdash \letInRLTm{y}{(\letInRLTm{x}{t}{u})}{u'} \thyeq \letInRLTm{x}{t}{\subst{(\wkTm{u'})}{u/y}} : \DiaTy{C}
    }\label{rule:dia-beta1/RL}%

    \inferrule[\RL/$\DiaTy$\nbhyp{}$\beta_2$]{ %
      \Gamma \vdash t : A\\
      \ExtCtx{\Gamma}{x : A} \vdash u : B
    }{%
      \Gamma \vdash \letInRLTm{x}{(\returnRLTm{t})}{u} \thyeq \returnRLTm{(\subst{u}{t/x})} : \DiaTy{B}
    }\label{rule:dia-beta2/RL}

  \end{mathpar}
  \caption{Well\hyp{}typed terms and equational theory for \RLC (omitting those of \STLC)}
  \label{fig:calculus-rlc}
\end{figure}

\noindent\emph{The calculus~\JLC}.
The well\hyp{}typed terms and equational theory for the modal fragment
of \JLC are defined in \cref{fig:calculus-jlc}.
The calculus \JLC extends \STLC with two constructs~$\letMapJLTm{}$
and~$\letJoinJLTm$, and five equations~\JL/$\DiaTy$\nbhyp{}$\eta$,
\JL/$\DiaTy$\nbhyp{}$\beta_1$, \JL/$\DiaTy$\nbhyp{}$\beta_2$,
\JL/$\DiaTy$\nbhyp{}com and \JL/$\DiaTy$\nbhyp{}ass.

Observe that the typing rule of the construct $\letMapJLTm$ in \JLC is
once again identical to $\letSLTm$ in \SLC.
As a result, axiom~$\SA$ can once again be derived in \JLC similar to
\SLC:
$$\EmptyCtx \vdash \lamTm{x}{\letMapInJLTm{y}{(\sndTm{x})}{(\pairTm{(\fstTm{x}),y})}} : A \ProdTy \DiaTy{B} \FunTy \DiaTy{(A \ProdTy B)}$$
Similarly, observe that the typing rule of the
construct~$\letJoinJLTm$ in \JLC is identical to construct~$\letLLTm$
in \MLC.
As a result, axiom $\JA$ can be derived in \JLC similar to
\MLC:
$$\EmptyCtx \vdash \lamTm{x}{\letJoinInJLTm{y}{x}{y}} : \DiaTy{\DiaTy{A}} \FunTy \DiaTy{A}$$
Axiom~$\RA$ cannot be derived in \JLC as there is no counterpart for
$\returnLLTm$ in \JLC.

A categorical model of \JLC is a cartesian\hyp{}closed category
equipped with a strong semimonad~$\DiaTy{}$.
We interpret the term construct~$\letMapJLTm{}$ (and in turn
axiom~$\SA$) in a categorical model~$\Cat[C]$ of \JLC, using the
tensorial strength of functor~$\DiaTy{}$ as before with \SLC and \RLC.
The interpretation of the term construct~$\letJoinJLTm{}$ (and in turn
axiom~$\JA$) is given by the semimonad structure of
functor~$\DiaTy{}$, which gives us a morphism~$\DiaTy{\DiaTy{X}}
\rightarrow \DiaTy{X}$ for all objects~$X$ in $\Mod[C]$.

\begin{proposition}[Categorical semantics for \JLC]\label{prop:cs:jlc}
  Given two terms~$t,u$ in \JLC, $\Gamma \vdash t \thyeq u : A$ if and
  only if for all categorical models~$\Cat[C]$ of \JLC $\eval{t} =
  \eval{u} : \eval{\Gamma} \rightarrow \eval{A} $ in $\Cat[C]$.
\end{proposition}

\begin{figure}[t]
  \begin{mathpar}
    \inferrule[\JL/$\DiaTy$\nbhyp{}Letmap]{%
      \Gamma  \vdash t : \DiaTy{A} \\ \ExtCtx{\Gamma}{x : A} \vdash u : B
    }{%
      \Gamma \vdash \letMapInJLTm{x}{t}{u} : \DiaTy{B}
    }\label{rule:letMapTm/JL}

   \inferrule[\JL/$\DiaTy$\nbhyp{}Let]{%
      \Gamma  \vdash t : \DiaTy{A} \\ \ExtCtx{\Gamma}{x : A} \vdash u : \DiaTy{B}
    }{%
      \Gamma \vdash \letJoinInJLTm{x}{t}{u} : \DiaTy{B}
    }\label{rule:letJoinTm/JL}

    \inferrule[\JL/$\DiaTy$\nbhyp{}$\eta$]{ %
      \Gamma \vdash t : \DiaTy{A}
    }{%
      \Gamma \vdash t \thyeq \letMapInJLTm{x}{t}{x} : \DiaTy{A}
    }\label{rule:dia-eta/JL}%

    \inferrule[\JL/$\DiaTy$\nbhyp{}$\beta_1$]{ %
      \Gamma \vdash t : \DiaTy{A} \\
      \ExtCtx{\Gamma}{x : A} \vdash u : B \\
      \ExtCtx{\Gamma}{y : B} \vdash u' : C
    }{%
      \Gamma \vdash \letMapInJLTm{y}{(\letMapInJLTm{x}{t}{u})}{u'} \thyeq \letMapInJLTm{x}{t}{\subst{(\wkTm{u'})}{u/y}} : \DiaTy{C}
    }\label{rule:dia-beta1/JL}%

    \inferrule[\JL/$\DiaTy$\nbhyp{}$\beta_2$]{ %
      \Gamma \vdash t : \DiaTy{A} \\
      \ExtCtx{\Gamma}{x : A} \vdash u : B \\
      \ExtCtx{\Gamma}{y : B} \vdash u' : \DiaTy{C}
    }{%
      \Gamma \vdash \letJoinInJLTm{y}{(\letMapInJLTm{x}{t}{u})}{u'} \thyeq \letJoinInJLTm{x}{t}{\subst{(\wkTm{u'})}{u/y}} : \DiaTy{C}
    }\label{rule:dia-beta2/JL}%

    \inferrule[\JL/$\DiaTy$\nbhyp{}com]{ %
      \Gamma \vdash t : \DiaTy{A} \\
      \ExtCtx{\Gamma}{x : A} \vdash u : \DiaTy{B} \\
      \ExtCtx{\Gamma}{y : B} \vdash u' : C
    }{%
      \Gamma \vdash \letMapInJLTm{y}{(\letJoinInJLTm{x}{t}{u})}{u'} \thyeq \letJoinInJLTm{x}{t}{(\letMapInJLTm{y}{u}{(\wkTm{u'})})} : \DiaTy{C}
    }\label{rule:dia-com/JL}%

   \inferrule[\JL/$\DiaTy$\nbhyp{}ass]{ %
      \Gamma \vdash t : \DiaTy{A} \\
      \ExtCtx{\Gamma}{x : A} \vdash u : \DiaTy{B} \\
      \ExtCtx{\Gamma}{y : B} \vdash u' : \DiaTy{C}
    }{%
      \Gamma \vdash \letJoinInJLTm{y}{(\letJoinInJLTm{x}{t}{u})}{u'} \thyeq \letJoinInJLTm{x}{t}{(\letJoinInJLTm{y}{u}{(\wkTm{u'})})} : \DiaTy{C}
    }\label{rule:dia-ass/JL}\\

  \end{mathpar}

  \caption{Well\hyp{}typed terms and equational theory for \JLC (omitting those of \STLC)}
  \label{fig:calculus-jlc}
\end{figure}

\section{Proof-relevant possible-world semantics}
\label{sec:prks}

In \cref{sec:overview}, possible\hyp{}world semantics was given for
the logic~\LL and its sublogics in a classical meta\hyp{}language
using sets and relations.
In this section, we will give proof\hyp{}relevant possible\hyp{}world
semantics for lax modal lambda calculi, for which we will instead work
in a constructive dependent type\hyp{}theory loosely based on the
proof assistant Agda.

We will use a type~$X : \Type$ in place of a set~$X$ and values~$x: X$
in place of elements~$x \in X$.
The arrow~$\to$ denotes functions, and quantifications~$\forall x$ and
$\Sigma_x$ denote universal and existential quantification
respectively, where $x : X$ is a value of some type~$X : \Type$ that
is left implicit.
A value of type~$\forall x.\, P(x)$ for some predicate~$ P : X \to
\Type$ is a function $\lambda x.\, p$ with $p : P(x)$.
When the expression~$p$ does not mention the variable~$x$ we will
leave the abstraction implicit and simply write~$p$ as a value of
$\forall x.\, P(x)$.
A value of type~$\Sigma_x.\, P(x)$ is a tuple~$\tuple{x,p}$, but we
will similarly leave the witness~$x$ implicit at times and write
$\tuple{\_,p}$ or simply $p$ for brevity.
\linebreak\linebreak
\noindent\emph{Semantics for \SLC}.
A proof\hyp{}relevant \SLC{}\hyp{}frame~$F = (W,\Ri,\Rm)$ is a triple
that consists of a \emph{type}~$W : \Type$ of worlds and two
proof\hyp{}relevant relations~$\Ri, \Rm : W \to W \to \Type$ with
\begin{itemize}
\item functions~$\reflRi : \forall w.\, w \Ri w$ and $\transRi :
  \forall w,w',w''.\, w \Ri w' \to w' \Ri w'' \to w \Ri w''$
  respectively proving the reflexivity and transitivity of $\Ri$ such
  that
  \begin{itemize}
    \item $\transRi{\reflRi,i}=i$ and $\transRi{i,\reflRi}=i$
    \item $\transRi{(\transRi{i,i'}),i''} = \transRi{i,(\transRi{i',i''}})$
    \end{itemize}
\item
  function~$\factor : \forall w,w',v.\, w \Ri w' \to w \Rm v \to
  \Sigma_{v'}.\, (w' \Rm v' \times v \Ri v')$ such that
  \begin{itemize}
    \item $\factor{\reflRi,m} = (m,\reflRi)$
    \item $\factor{(\transRi{i_1,i_2}),m} = (m_2',(\transRi{i_1',i_2'}))$
      \\ where $(m_1',i_1') = \factor{i_1,m}$ and $(m_2',i_2') = \factor{i_2,m_1'}$.
  \end{itemize}
\item function~$\incl : \forall w, v.\, w \Rm v \to w \Ri v$ such
  that
  \begin{itemize}\item $\transRi{i,(\incl{m'})} = \transRi{(\incl{m}),i'}$, where $(i',m') = \factor{i,m}$ \end{itemize}
\end{itemize}

The function~$\reflRi$ and $\transRi$ are the proof\hyp{}relevant
encoding of reflexivity and transitivity of $\Ri$ respectively.
These functions are subject to the accompanying coherence laws, which
state that the proof computed by $\reflRi$ must be the unit of
$\transRi$, i.e. $\Ri$ must form a category~$\WCat$.
The coherence laws facilitate a sound interpretation of \SLC's
equational theory.

The functions~$\factor$ and~$\incl$ are proof\hyp{}relevant encodings
of the forward confluence (${\Ri^{-1} ; \Rm} \subseteq {\Rm ;
  \Ri^{-1}}$) and inclusion (${\Rm} \subseteq {\Ri}$) conditions
respectively.
Given a proof of $w \Ri w' $ (i.e. $w' \Ri^{-1} w$) and $w \Rm v$,
$\factor$ returns a pair of proofs for some world~$v'$: $w' \Rm v'$
and $v \Ri v'$ (i.e. $v' \Ri^{-1} v$).
Similarly, given a proof of $w \Rm v $, $\incl$ returns a proof of $w
\Ri v $.
These functions are also accompanied by the stated coherence laws.

The proof\hyp{}relevant relation~$\Ri$ in a \SLC{}\hyp{}frame
determines a category~$\WiCat$ whose objects are given by worlds and
morphisms by proofs of $\Ri$, with $\reflRi$ witnessing the identity
morphisms and $\transRi$ witnessing the composition of morphisms.
This determines a category~$\Psh[\WCat]$ of covariant presheaves
indexed by~$\WCat$.
The objects in the category~$\Psh[\WCat]$ are presheaves and the
morphisms are natural transformations.
A presheaf~$P$ is given by a family of meta\hyp{}language types~$P_w :
\Type$ indexed by worlds~$w : W$, accompanied by ``transportation''
functions~$\wkPsh{P}{}{} : \forall w,w'.\, w \Ri w' \to P_w \to
P_{w'}$ subject to the ``functoriality'' conditions that
$\wkPsh{P}{\reflRi}{p} = p$ and $\wkPsh{P}{(\transRi{i,i'})}{p} =
\wkPsh{P}{i'}{(\wkPsh{P}{i}{p})}$, for arbitrary values~$i : w \Ri
w'$, $i' : w' \Ri w''$ and $p : P_w$.
A natural transformation~$f : P \natto Q$ is a family of functions
$\forall w.\, P_w \to Q_w$ subject to a ``naturality'' condition that
$f\,(\wkPsh{P}{i}{p}) = \wkPsh{Q}{i}{(f\, p)}$.
\begin{proposition}[$\DiaFun$ Strong Functor]\label{prop:psh-dia-strong-fun}
  The presheaf category~$\Psh[\WCat]$ determined by a
  \SLC{}\hyp{}frame exhibits a strong endofunctor~$(\DiaFun{P})_w =
  \Sigma_v.\, w \Rm v \times P_{v}$ for some world~$w$ and
  presheaf~$P$.
\end{proposition}
\begin{proof}
  We show that the family~$\DiaFun{P}$ is a presheaf by defining the
  function~$\wkPsh{\DiaFun{P}}{}{}$ as:
  $$\wkPsh{\DiaFun{P}}{i}{\tuple{v,m,p}} = \tuple{v',m',\wkPsh{P}{i'}{p}}
  \qquad \text{where}\ \tuple{v',m' : w' \Rm v',i' : v \Ri v'} =
  \factor{i,m}$$
  The functoriality conditions follow from the coherence conditions on
  $\factor$.
  To show the operator~$\DiaFun$ is a functor on presheaves, we must
  show that for every natural transformation~$f : P \natto Q$ there
  exists a natural transformation~$\DiaFun{f} : \DiaFun{P} \natto
  \DiaFun{Q}$.
  This natural transformation is defined as $\DiaFun{f}\,\tuple{v,m,p}
  = \tuple{v,m,(f_v\,p)}$ by applying $f$ at the world~$v$ witnessing
  the $\Sigma$ quantification.
  The laws of the functor~$\DiaFun$ follow immediately.
  The strength of the functor~$\DiaFun$ can be defined using the
  function~$\incl$ and its accompanying coherence conditions.
\end{proof}

\cref{prop:cs:slc,prop:psh-dia-strong-fun} give us that~$\Psh[\WCat]$
is a categorical model of \SLC.
For clarity, we elaborate on this consequence by giving a direct
interpretation of \SLC in $\Psh[\WCat]$.

A proof\hyp{}relevant possible\hyp{}world \emph{model}~$M = (F,V)$
couples a proof\hyp{}relevant frame~$F$ with a valuation function~$V$
that assigns to a base type~$\BaseTy$ a presheaf~$V_\BaseTy : W \to
\Type$.
Given such a model, the types in \SLC are interpreted as presheaves,
i.e. we interpret a type~$A$ as a family~$\eval{A}_w : \Type$ indexed
by an arbitrary world~$w : W$---as shown on the left below.
\vspace{-\abovedisplayskip}
\begin{center}
\begin{minipage}{.6\linewidth}
\centering
\begin{equation*}
  \begin{array}{>{\evallbracket}l@{\evalrbracket}l @{\;}c@{\;} l}
    \BaseTy    & _w & = & V_{\BaseTy,w} \\
    \UnitTy    & _w & = & \top \\
    A \ProdTy B & _w & = & \eval{A}_{w} \times \eval{B}_{w}\\
    A \FunTy B & _w & = & \forall w'.\, w \Ri w' \to \eval{A}_{w'} \to \eval{B}_{w'}\\
    \DiaTy{A}  & _w & = & \textstyle\sum\nolimits_{v}.\, w \Rm v  \times \eval{A}_{v} \\
  \end{array}
\end{equation*}
\end{minipage}%
\begin{minipage}{.4\linewidth}
\centering
\begin{equation*}
  \begin{array}{>{\evallbracket}l@{\evalrbracket}l @{\;}c@{\;} l}
   \EmptyCtx    & _w & = & \top \\
   \ExtCtx{\Gamma}{x : A} & _w & = & \eval{\Gamma}_{w} \times \eval{A}_{w}\\    
  \end{array}
\end{equation*}
\end{minipage}
\end{center}
The interpretation of the base type~$\BaseTy$ is given by the
valuation function~$V$, and the unit, product and function types are
interpreted as usual using their semantic counterparts.
We interpret the $\DiaTy{}$ modality using the
proof\hyp{}relevant quantifier~$\textstyle\sum$: the interpretation of
a type~$\DiaTy{A}$ at a world~$w$ is given by the interpretation
of~$A$ at some modal future world~$v$ along with a proof of $w \Rm v$
witnessing the connection from~$w$ to $v$ via~$\Rm$.
The typing contexts are interpreted as usual by taking the cartesian
product of presheaves.

The terms in \SLC are interpreted as natural transformations by
induction on the typing judgment.
Interpretation of \STLC terms follows the usual routine:
we interpret variables by projecting the environment~$\gamma :
\eval{\Gamma}_w$ using a function~$\lookup$, the unit and pair
constructs~($\unitTm$, $\pairTm$, $\fstTm$, $\sndTm$) with their
semantic counterparts~($()$, $\pair{-}{-}$, $\pi_1$, $\pi_2$), and the
function constructs~($\lamLabel$, $\appTm$) with appropriate semantic
function abstraction and application.
\begin{align*}
  \begin{array}{>{\evallbracket\ }l@{\ \evalrbracket} @{\$} l @{\;}c@{\;} l}
    \multicolumn{4}{l}{\eval : \Gamma \vdash A \to (\forall w.\, \eval{\Gamma}_{w} \to  \eval{A}_{w})} \\
    x               & \gamma & = & \lookup{x,\gamma} \\
    \unitTm         & \gamma & = & () \\
    \pairTm{t,u}    & \gamma & = & \pair{\eval{t}\$ \gamma}{\eval{u}\$ \gamma} \\
    \fstTm{t}    & \gamma & = & \pi_1{(\eval{t}\$ \gamma)} \\
    \sndTm{t}    & \gamma & = & \pi_2{(\eval{t}\$ \gamma)} \\
    \lamTm{x}{t}       & \gamma & = & \lambda i.\, \lambda a.\, \eval{t}\$\pair{\wkPsh{\eval{\Gamma}}{i}{\gamma}}{a} \\
    \appTm{t,u}     & \gamma & = & (\eval{t}\$ \gamma)\$ \reflRi\$ (\eval{u}\$ \gamma) \\
    \letInSLTm{x}{t}{u} & \gamma & = & \pair{m}{\eval{u}\$ \pair{\wkPsh{\eval{\Gamma}}{(\incl{m})}{\gamma}}{a}} \\
    \multicolumn{4}{l}{\quad\quad \text{where } \pair{m : w \Rm v}{a : \eval{A}_v} = \eval{t}\$ \gamma}
  \end{array}
\end{align*}
The interesting case is that of $\letSLTm$: given terms~$\Gamma \vdash
t : \DiaTy{A}$ and $\ExtCtx{\Gamma}{x : A} \vdash u : B$, and an
environment~$\gamma : \eval{\Gamma}_w$, we must produce an element of
type~$\eval{\DiaTy{B}}_w = \textstyle\sum_v.\, w \Rm v \times \eval{B}_{v}$.
Recursively interpreting~$t$ gives us a pair~$\pair{m : w \Rm v}{a :
  \eval{A}_v}$, using the former of which we transport~$\gamma$ along
$\Rm$ to the world~$v$, as $\wkPsh{\eval{\Gamma}}{(\incl{m})}{\gamma}
: \eval{\Gamma}_v$, which is in turn used to recursively
interpret~$u$, thus obtaining the desired element of
type~$\eval{B}_v$.
\linebreak\linebreak
\noindent\emph{Semantics for \RLC}.
A proof\hyp{}relevant \RLC{}\hyp{}frame~$(W,\Ri,\Rm)$ is a
\SLC{}\hyp{}frame that exhibits:
\begin{itemize}
  \item function~$\reflRm : \forall w.\, w \Rm w$, such that
    \begin{itemize}
    \item $\factor{i,\reflRm} = \tuple{\reflRm,i}$
    \item $\incl{\reflRm}=\reflRi$
    \end{itemize}
  \end{itemize}
\begin{proposition}[$\DiaFun$ Strong Pointed]\label{prop:psh-dia-strong-pointed}
   The strong functor~$\DiaFun$ on the category of
   presheaves~$\Psh[\WCat]$ determined by a \RLC{}\hyp{}frame is
   strong pointed.
\end{proposition}
\begin{proof}
  To show that $\DiaFun$ is pointed, we define $\point : P \natto
  \DiaFun{P}$ using function~$\reflRm$, and then use the coherence
  law~$\incl{\reflRm}=\reflRi$ to show that $\point$ is a strong
  natural transformation.
\end{proof}

\cref{prop:cs:rlc,prop:psh-dia-strong-pointed} give us
that~$\Psh[\WCat]$ is a categorical model of \RLC for
\RLC{}\hyp{}frames.
The interpretation of the modal fragment of \RLC can be given
explicitly in $\Psh[\WCat]$ as:
\begin{equation*}
  \begin{array}{>{\evallbracket\ }l@{\ \evalrbracket} @{\$} l @{\;}c@{\;} l}
    \returnRLTm{t}    & \gamma & = & \pair{\reflRm}{\eval{t}\$ \gamma} \\
    \letInRLTm{x}{t}{u} & \gamma & = & \pair{m}{\eval{u}\$ \pair{\wkPsh{\eval{\Gamma}}{(\incl{m})}{\gamma}}{a}} \\
    \multicolumn{4}{l}{\quad\quad \text{where } \pair{m : w \Rm v}{a : \eval{A}_v} = \eval{t}\$ \gamma}
  \end{array}
\end{equation*}
\noindent\emph{Semantics for \JLC}.
A proof\hyp{}relevant \JLC{}\hyp{}frame~$(W,\Ri,\Rm)$ is a
\SLC{}\hyp{}frame that exhibits:
\begin{itemize}
  \item
    function~$\transRm : \forall u, v, w.\ u \Rm v \to v \Rm w \to u
    \Rm w$, such that
    \begin{itemize}
    \item
      $\factor{i,(\transRm{m_1,m_2})} = \tuple{\transRm{m_1',m_2'}, i_2'}$
      \\ where $(m_1',i_1') = \factor{i,m_1}$ and $(m_2',i_2') = \factor{i_1',m_2}$.
    \item
      $\transRm{(\transRm{m_1,m_2}),m_3} = \transRm{m_1,(\transRm{m_2,m_3})}$
      \item $\incl{(\transRm{m_1,m_2})}=\transRi{(\incl{m_1}),(\incl{m_2})}$
    \end{itemize}
\end{itemize}

\begin{proposition}[$\DiaFun$ Strong Semimonad]\label{prop:psh-dia-strong-semimonad}
  The strong functor~$\DiaFun$ on the category of
   presheaves~$\Psh[\WCat]$ determined by a \JLC{}\hyp{}frame is
  a strong semimonad.
\end{proposition}
\begin{proof}
  We define $\join : \DiaFun{\DiaFun{P}} \natto \DiaFun{P}$ using the
  function~$\transRm$ to show $\DiaFun$ is a semimonad, and then use
  the coherence law~$\incl{(\transRm{m_1,m_2})} =
  \transRi{(\incl{m_1}),(\incl{m_2})}$ to show that $\join$ is a
  strong natural transformation---giving us that $\join$ is a strong
  semimonad.
\end{proof}

\cref{prop:cs:jlc,prop:psh-dia-strong-semimonad} give us
that~$\Psh[\WCat]$ is a categorical model of \JLC for
\JLC{}\hyp{}frames.
The interpretation of the modal fragment of \JLC can be given
explicitly in $\Psh[\WCat]$ as:
\begin{equation*}
  \begin{array}{>{\evallbracket\ }l@{\ \evalrbracket} @{\$} l @{\;}c@{\;} l}
    \letMapInJLTm{x}{t}{u} & \gamma & = & \pair{m}{\eval{u}\$ \pair{\wkPsh{\eval{\Gamma}}{(\incl{m})}{\gamma}}{a}} \\
    \multicolumn{4}{l}{\quad\quad \text{where } \pair{m : w \Rm v}{a : \eval{A}_v} = \eval{t}\$ \gamma}\\
    \letJoinInJLTm{x}{t}{u} & \gamma & = & \pair{\transRm{m,m'}}{b} \\
     \multicolumn{4}{l}{\quad\quad \text{where } \pair{m : w \Rm v}{a : \eval{A}_v} = \eval{t}\$ \gamma}\\
     \multicolumn{4}{l}{\quad\quad\quad\quad\quad \pair{m' : v \Rm v'}{b : \eval{B}_{v'}} = \eval{u}\$ \pair{\wkPsh{\eval{\Gamma}}{(\incl{m})}{\gamma}}{a}}\\
  \end{array}
\end{equation*}

\noindent\emph{Semantics for \MLC}.
A proof\hyp{}relevant \MLC{}\hyp{}frame~$F = (W,\Ri,\Rm)$ is both a
\RLC{}\hyp{}frame and \JLC{}\hyp{}frame that further exhibits the unit
laws~$\transRm{\reflRm,m}=m$ and $\transRm{m,\reflRm}=m$.
That is, proofs of $\Rm$ now form a category $\WmCat$ with a functor
$\WmCat \to \WiCat$ given by function~$\incl$.

\begin{proposition}[$\DiaFun$ Strong Monad]\label{prop:psh-dia-strong-monad}
    The strong functor~$\DiaFun$ on the category of
    presheaves~$\Psh[\WCat]$ determined by a \MLC{}\hyp{}frame is a
    strong monad.
\end{proposition}
\begin{proof}
  We apply
  \cref{prop:psh-dia-strong-fun,prop:psh-dia-strong-pointed,prop:psh-dia-strong-semimonad}
  to show that the functor~$\DiaFun$ is a strong pointed semimonad.
  We then use the unit laws of the category~$\WmCat$ to prove the unit
  laws of the monad~$\DiaFun$.
\end{proof}

\cref{prop:cs:mlc,prop:psh-dia-strong-monad} give us
that~$\Psh[\WCat]$ is a categorical model of \MLC for
\MLC{}\hyp{}frames.
The interpretation of the modal fragment of \MLC can be given
explicitly in $\Psh[\WCat]$ as:
\begin{equation*}
  \begin{array}{>{\evallbracket\ }l@{\ \evalrbracket} @{\$} l @{\;}c@{\;} l}
    \returnLLTm{t}    & \gamma & = & \pair{\reflRm}{\eval{t}\$ \gamma} \\
    \letInLLTm{x}{t}{u} & \gamma & = & \pair{\transRm{m,m'}}{b} \\
     \multicolumn{4}{l}{\quad\quad \text{where } \pair{m : w \Rm v}{a : \eval{A}_v} = \eval{t}\$ \gamma}\\
     \multicolumn{4}{l}{\quad\quad\quad\quad\quad \pair{m' : v \Rm v'}{b : \eval{B}_{v'}} = \eval{u}\$ \pair{\wkPsh{\eval{\Gamma}}{(\incl{m})}{\gamma}}{a}}\\
  \end{array}
\end{equation*}

\begin{theorem}[Soundness of proof\hyp{}relevant possible\hyp{}world semantics]
  For any two terms~$\Gamma \vdash t,u : A$ in
  \SLC{}/\RLC{}/\JLC{}/\MLC{}, if $\Gamma \vdash t \thyeq u : A$ then
  $\eval{t} = \eval{u}$ for an arbitrary proof\hyp{}relevant
  possible\hyp{}world model determined by the respective
  \SLC{}/\RLC{}/\JLC{}/\MLC{}\hyp{}frames.
\end{theorem}
\begin{proof}
  
  Applying
  \cref{prop:psh-dia-strong-fun,prop:psh-dia-strong-pointed,prop:psh-dia-strong-semimonad,prop:psh-dia-strong-monad}
  to \cref{prop:cs:slc,prop:cs:rlc,prop:cs:jlc,prop:cs:mlc}
  accordingly gives us that the category~$\Psh[\WCat]$ determined by a
  \SLC{}/\RLC{}/\JLC{}/\MLC{}\hyp{}frame is a categorical model of the
  respective calculus.
  As a result, we get the soundness of the equational theory for
  possible\hyp{}world models via soundness of the equational theory
  for categorical models.
\end{proof}

\section{Normalization, completeness and inadmissibility results}
\label{sec:nci}

Catarina Coquand~\cite{Coquand93,Coquand02} proved normalization for
\STLC in the proof assistant Alf~\cite{MagnussonN93} by constructing
an instance of Mitchell and Moggi's proof\hyp{}relevant
possible\hyp{}world semantics.
This model\hyp{}based approach to normalization, known as
Normalization by Evaluation~(NbE)~\cite{BergerS91,BergerES98},
dispenses with tedious syntactic reasoning that typically complicate
normalization proofs.
In this section, we extend Coquand's result to lax modal lambda
calculi and observe corollaries including completeness and
inadmissibility of irrelevant modal axioms.

The objective of NbE is to define a function~$\norm : \Tm{\Gamma}{A}
\to \Nf{\Gamma}{A}$, assigning a \emph{normal form} to every term in
the calculus.
We write $\Tm{\Gamma}{A}$ to denote all terms~$\Gamma \vdash t : A$
and $\Nf{\Gamma}{A}$ to denote all normal
forms~$\Gamma \vdashNf n : A$.
The normal form judgments~$\Gamma \vdashNf n : A$ are defined in
\cref{fig:slc-nf} alongside \emph{neutral term} judgments~$\Gamma
\vdashNe n : A$, which can be roughly understood as
``straight\hyp{}forward'' inferences that do not involve introduction
rules.
\begin{figure}[t]
  \begin{mathpar}
    \inferrule[Ne/Var]{%
      \Gamma \vdashVar x : A
    }{%
      \Gamma \vdashNe x : A
    }%

    \inferrule[Nf/Up]{%
      \Gamma \vdashNe n : \BaseTy
    }{%
      \Gamma \vdashNf \upTm{n} : \BaseTy
    }%

    \inferrule[Nf/Unit]{%
    }{%
      \Gamma \vdashNf \unitTm : \UnitTy
    }%
    
    \inferrule[Ne/$\ProdTy$\nbhyp{}Elim\nbhyp{}1]{%
      \Gamma \vdashNe n : A \ProdTy B
    }{%
      \Gamma \vdashNe \fstTm{n} : A
    }%

    \inferrule[Ne/$\ProdTy$\nbhyp{}Elim\nbhyp{}2]{%
      \Gamma \vdashNe n : A \ProdTy B
    }{%
      \Gamma \vdashNe \sndTm{n} : B
    }%

    \inferrule[Nf/$\ProdTy$\nbhyp{}Intro]{%
      \Gamma \vdashNf n : A\\
      \Gamma \vdashNf m : B
    }{%
      \Gamma \vdashNf \pairTm{n,m} : A \ProdTy B
    }%

    \inferrule[Nf/$\FunTy$\nbhyp{}Intro]{%
      \ExtCtx{\Gamma}{x : A} \vdashNf n : B
    }{%
      \Gamma \vdashNf \lamTm{x}{n} : A \FunTy B
    }%

    \inferrule[Ne/$\FunTy$\nbhyp{}Elim]{%
      \Gamma \vdashNe n : A \FunTy B\\
      \Gamma \vdashNf m : A
    }{%
      \Gamma \vdashNe \appTm{n,m} : B
    }%

    \inferrule[NF/$\DiaTy$\nbhyp{}Letmap/\SL]{%
      \Gamma  \vdashNe n : \DiaTy{A} \\
      \ExtCtx{\Gamma}{x : A} \vdashNf m : B
    }{%
      \Gamma \vdashNf \letInSLTm{x}{n}{m} : \DiaTy{B}
    }

  \end{mathpar}
  \caption{Neutral terms and Normal forms for \SLC}
  \label{fig:slc-nf}
\end{figure}

To define $\norm$ for \SLC, we construct a possible\hyp{}world
model~$\tuple{N,V}$, known as the NbE model, with a \SLC{}\hyp{}frame
$N = \triple{\Ctx}{\Wk{}{}}{\RSL{}}$ consisting of contexts for
worlds, the context inclusion relation $\Wk{}{}$ for $\Ri$, and the
accessibility relation $\RSL{}{}$ for $\Rm$.
The valuation is given by neutral terms as
$V_{\BaseTy,\Gamma} = \Ne{\Gamma}{\BaseTy}$.
The relation~$\RSL{}{}$ can be defined inductively as below.
This definition states that $\RSL{\Gamma}{\Delta}$ if and only if
$\Delta = \Gamma , x : A$ for some variable~$x$ (not in $\Gamma$) and
type~$A$ such that there exists a neutral term~$\Gamma \vdash n :
\DiaTy{A}$.
\begin{mathpar}
  \inferrule[]{%
    \Gamma  \vdashNe n : \DiaTy{A} \quad (x\ \text{not in}\ \Gamma)
  }{%
    \varExtOnce{n} : \RSL{\Gamma}{(\ExtCtx{\Gamma}{x : A})}
  }
\end{mathpar}
The proof\hyp{}relevant relation~$\RSL{}{}$ is neither reflexive nor
transitive, but is included in the relation~$\Wk{}{}$ since we can
define a function $\incl : \forall \Gamma, \Delta.\,
\RSL{\Gamma}{\Delta} \to \Wk{\Gamma}{\Delta}$.
We can also show that the \SLC{}\hyp{}frame~$N$ satisfies the forward
confluence condition by defining a function~$\factor : \forall
\Gamma,\Gamma',\Delta.\, \Wk{\Gamma}{\Gamma'} \to \RSL{\Gamma}{\Delta}
\to \exists \Delta'.\, (\RSL{\Gamma'}{\Delta'} \times
\Wk{\Delta}{\Delta'})$.

By construction, we obtain an interpretation of terms~$\eval :
\Tm{\Gamma}{A} \to (\forall \Delta .\, \eval{\Gamma}_\Delta \to
\eval{A}_\Delta)$ in the NbE model as an instance of the generic
interpreter for an arbitrary possible\hyp{}world model
(\cref{sec:prks}).
This model exhibits two type\hyp{}indexed functions characteristic of
NbE models known as $\reify$ and $\reflect$, which are defined for the
modal fragment as follows:
\begin{align*}
  & \begin{array}{l @{\$} l @{\;}c@{\;} l}
      \multicolumn{4}{l}{\reify_{A} : \forall\, \Gamma.\, \eval{A}_\Gamma \to \Gamma \vdashNf A} \\
    \ldots \\
    \reify_{\DiaTy{A};\Gamma}   & \pair{\varExtOnce{n} : \RSL{\Gamma}{(\ExtCtx{\Gamma}{x : B})}}{a : \eval{A}_{\ExtCtx{\Gamma}{x : B}}} & = & \letInSLTm{x}{n}{(\reify_{A;(\ExtCtx{\Gamma}{x : B})}{a})}
    \end{array}\\[\the\abovedisplayskip]
  & \begin{array}{l @{\$} l @{\;}c@{\;} l}
      \multicolumn{4}{l}{\reflect_{A} : \forall\, \Gamma.\, \Gamma \vdashNe A \to \eval{A}_\Gamma}\\
    \ldots \\
    \reflect_{\DiaTy{A};\Gamma} & n& = & \pair{\varExtOnce{n}}{\reflect_{A;(\ExtCtx{\Gamma}{x : A})}{x}}
    \end{array}
\end{align*}
The function $\reify$ is a type\hyp{}indexed natural transformation,
which for the case of type~$\DiaTy{A}$ in some context $\Gamma$, is
given as argument an element of type $\eval{\DiaTy{A}}_\Gamma$, which
is $\Sigma_\Delta. \RSL{\Gamma}{\Delta} \times \eval{A}_\Delta$.
The first component gives us a neutral~$\Gamma \vdashNe n :
\DiaTy{B}$, and recursively reifying the second component gives us a
normal form of~$\ExtCtx{\Gamma}{x : B} \vdashNf
\reify_{A;(\ExtCtx{\Gamma}{x : B})}{a} : A$.
We use these to construct the normal form $\Gamma \vdashNf
\letInSLTm{x}{n}{(\reify_{A;(\ExtCtx{\Gamma}{x : B})}{a})} : \DiaTy{A}$,
which is the desired result.
The function $\reflect$, on the other hand, constructs a value pair of
type~$\eval{\DiaTy{A}}_\Gamma = \Sigma_\Delta. \RSL{\Gamma}{\Delta}
\times \eval{A}_\Delta$ using the given neutral~$\Gamma \vdashNe n :
\DiaTy{A}$ and picking $\ExtCtx{\Gamma}{x : A}$ (with a fresh
variable~$x$ not in $\Gamma$) as the witness of $\Delta$ to obtain a
value of type $\eval{A}_{\ExtCtx{\Gamma}{x : A}}$ by reflecting the
the variable~$x$ as a neutral term~$\ExtCtx{\Gamma}{x : A} \vdashNe x
: A$.
These functions are key to defining~$\quote : (\forall \Delta .\,
\eval{\Gamma}_\Delta \to \eval{A}_\Delta) \to \Nf{\Gamma}{A}$, which
in turn gives us the function~$\norm$:
$$\norm{t} = \quote{\eval{t}}$$

\begin{figure}
\begin{mathpar}
  \inferrule[\RL/NF/$\DiaTy$\nbhyp{}Return]{%
    \Gamma \vdashNf n : A }{%
    \Gamma \vdashNf \returnRLTm{n} : \DiaTy{A}
  }

  \inferrule[\RL/NF/$\DiaTy$\nbhyp{}Letmap]{%
    \Gamma  \vdashNe n : \DiaTy{A} \\
    \ExtCtx{\Gamma}{x : A} \vdashNf m : B }{%
    \Gamma \vdashNf \letInRLTm{x}{n}{m} : \DiaTy{B}
  }
  
    \inferrule[\JL/NF/$\DiaTy$\nbhyp{}Letmap]{%
    \Gamma  \vdashNe n : \DiaTy{A} \\
    \ExtCtx{\Gamma}{x : A} \vdashNf m : B
  }{%
    \Gamma \vdashNf \letMapInJLTm{x}{n}{m} : \DiaTy{B}
  }

  \inferrule[\JL/NF/$\DiaTy$\nbhyp{}Let]{%
    \Gamma  \vdashNe n : \DiaTy{A} \\
    \ExtCtx{\Gamma}{x : A} \vdashNf m : \DiaTy{B}
  }{%
    \Gamma \vdashNf \letJoinInJLTm{x}{n}{m} : \DiaTy{B}
  }

  \inferrule[\LL/NF/$\DiaTy$\nbhyp{}Return]{%
    \Gamma  \vdashNf n : A
  }{%
    \Gamma \vdashNf \returnLLTm{n} : \DiaTy{A}
  }

  \inferrule[\LL/NF/$\DiaTy$\nbhyp{}Let]{%
    \Gamma  \vdashNe n : \DiaTy{A} \\
    \ExtCtx{\Gamma}{x : A} \vdashNf m : \DiaTy{B}
  }{%
    \Gamma \vdashNf \letInLLTm{x}{n}{m} : \DiaTy{B}
  }

\end{mathpar}
\caption{Normal forms for modal fragments of \RLC, \JLC and \MLC}
\label{fig:rjllc-nf}
\end{figure}

NbE models can be constructed likewise for the calculi~\RLC,~\JLC and~\MLC.
The normal forms of these calculi are defined in \cref{fig:rjllc-nf}.
To construct the model, we uniformly pick contexts for worlds, the
relation~$\Wk{}{}$ for $\Ri$, and the respective modal accessibility
relation defined in \cref{fig:rjllc-acc} for $\Rm$. As before, we also
pick neutrals terms for valuation.
\begin{figure}
\begin{mathpar}

  \inferrule[]{%
  }{%
    \nilExt : \RRL{\Gamma}{\Gamma}
  }

  \inferrule[]{%
    \Gamma \vdashNe n : \DiaTy{A}\\ (x\ \text{not in}\ \Gamma)
  }{%
    \varExtOnce{n} : \RRL{\Gamma}{\Gamma,x : A}
  }\\

  \inferrule[]{%
    \Gamma \vdashNe n : \DiaTy{A}\\ (x\ \text{not in}\ \Gamma)
  }{%
    \varExtOnce{n} : \RJL{\Gamma}{\Gamma,x : A} }

  \inferrule[]{%
    \Gamma  \vdashNe n : \DiaTy{A} \\
    m : \RJL{\Gamma,x : A}{\Delta}
  }{%
    \varExt{n,m} : \RJL{\Gamma}{\Delta}
  }\\

  \inferrule[]{%
  }{%
    \nilExt : \RLL{\Gamma}{\Gamma}
  }

  \inferrule[]{%
    \Gamma  \vdashNe n : \DiaTy{A} \\
    m : \RLL{\Gamma,x : A}{\Delta}
  }{%
    \varExt{n,m} : \RLL{\Gamma}{\Delta}
  }
\end{mathpar}
\caption{Modal accessibility relations for \RLC, \JLC and \MLC}
\label{fig:rjllc-acc}
\end{figure}
Observe that relation~$\RLL{}{}$ satisfies the inclusion condition (we
can define function~$\incl$) and is reflexive (we can define
$\reflRm$) and transitive (we can define $\transRm$).
On the other hand, the relations~$\RRL{}{}$ and $\RJL{}{}$ both
satisfy the inclusion condition and are respectively reflexive and
transitive, but not the other way around.
The main idea behind the definitions of these relations is that they
imitate the binding structure of the normal forms in
\cref{fig:rjllc-nf}.

\begin{theorem}[Correctness of normalization]\label{prop:norm:slc-correct}
  For all terms~$\Gamma \vdash t : A $ in \SLC{}/\RLC{}/\JLC{}/\MLC{},
  there exists a normal form~$\Gamma \vdashNf n : A $ such that $t
  \thyeq n$.
\end{theorem}
\begin{proof}
 By virtue of the function~$\norm$, we get that every term~$t$ has a
 normal form~$\norm{t}$. Using a standard logical relation based
 argument we can further show that $t \thyeq \norm{t}$.
\end{proof}

\begin{corollary}[Completeness of proof\hyp{}relevant possible\hyp{}world semantics]
  For any two terms~$\Gamma \vdash t,u : A$ in
  \SLC{}/\RLC{}/\JLC{}/\MLC{}, if $\eval{t} = \eval{u}$ in all
  proof\hyp{}relevant possible\hyp{}world models determined by the
  respective \SLC{}/\RLC{}/\JLC{}/\MLC{}\hyp{}frames, then $\Gamma
  \vdash t \thyeq u : A$.
\end{corollary}
\begin{proof}
  In the respective NbE model, we know $\eval{t} = \eval{u}$ implies
  $\norm{t} = \norm{u}$ by definition of $\norm$. By
  \cref{prop:norm:slc-correct}, we also know $t \thyeq \norm{t}$ and
  $u \thyeq \norm{u}$, thus~$t \thyeq u$.
\end{proof}

\begin{corollary}[Inadmissibility of irrelevant modal axioms]
  The axiom~\RA is not derivable in \SLC or \JLC, and similarly the
  axiom~\JA is not derivable in \SLC or \RLC.
\end{corollary}
\begin{proof}
  We first observe that for any neutral term~$\Gamma \vdashNe n : A$,
  the type~$A$ is a subformula of some type in context~$\Gamma$.
  We then show by case analysis that there cannot exist a derivation
  of the judgment~$\EmptyCtx \vdashNf A \FunTy \DiaTy{A}$ in \SLC or
  \JLC, and thus there cannot exist a derivation of axiom~\RA in
  either calculus---because every term must have a normal form, as
  shown by the normalization function.
  A similar argument can be given for axiom~\JA in \SLC and \RLC.
\end{proof}
  
\section{Related and further work}
\label{sec:related}

Simpson~\cite[Chapter 3]{Simpson94a} gives a comprehensive summary of
several IMLs alongside a detailed discussion of their characteristic
axioms and possible\hyp{}world semantics.
Notable early work on IMLs can be traced back to
Fischer-Servi~\cite{Servi77,Servi81}, Božić and Došen~\cite{BozicD84},
Sotirov~\cite{Sotirov80}, Plotkin and Stirling~\cite{PlotkinS86},
Wijesekera~\cite{Wijesekera90}, and many others since.

\emph{Global vs local interpretation}.
Fairtlough and Mendler~\cite{FairtloughM97} give a different
presentation of \LL.
The truth of their lax modality~$\CircTy$ is defined ``globally'' as
follows:
\begin{equation*}
  \begin{array}{l@{\;\Vdash\;} @{\;} l @{\;\text{iff}\;}c@{\;} l}
    \Mod[M],w & \CircTy{A} & & \text{for all}\ w'\ \text{s.t.}\ w \Ri w',\ \text{there exists}\ v\ \text{with}\ w' \Rm v\  \text{and}\ \Mod[M],v \Vdash A
  \end{array}
\end{equation*}
Their semantics does not require the forward confluence
condition~${\Ri^{-1} ; \Rm} \subseteq {\Rm ; \Ri^{-1}}$ since
monotonicity follows immediately the definition of the satisfaction
relation.
In the presence of forward confluence, this definition is equivalent
to the ``local'' one we have chosen in \cref{sec:overview} for the
$\DiaTy$ modality~\cite{Simpson94a,GrootSC25}, which means
$\CircTy{A}$ is true if and only if $\DiaTy{A}$ is true.
This observation can also be extended to the respectively determined
presheaf functors:
\begin{proposition}
  The presheaf functors~$\DiaFun$ and $\CircFun$ are naturally
  isomorphic.
\end{proposition}

In modal logic, the forward confluence condition forces the
axiom~$\DiaTy{(A \lor B)} \implies \DiaTy{A} \lor \DiaTy{B}$ to be
true~\cite{BalbianiGGO24}, which does not hold generally for strong
functors.
This observation, however, presupposes that the satisfaction clause
for the disjunction connective is defined as follows:
\begin{equation*}
  \begin{array}{l@{\;\Vdash\;} @{\;} l @{\;\text{iff}\;}c@{\;} l}
    \Mod[M],w & A \lor B & & \Mod[M],w \Vdash A\ \text{or}\ \Mod[M],w \Vdash B
  \end{array}
\end{equation*}
This ``Kripke\hyp{}style'' interpretation of disjunction is not
suitable for our purposes given that our objective is to
constructively prove completeness for lambda calculi using
possible\hyp{}world semantics.
Completeness in the presence of sum types in lambda calculi is a
notorious matter~\cite{AltenkirchDHS01,FioreS99} that requires further
investigation in the presence of the lax modality.

\emph{Box modality in lax logic}. Fairtlough and Mendler~\cite{FairtloughM97} note
that ``there is no point'' in defining a $\BoxTy{}$ modality for \LL
since it ``yields nothing new''.
With the following standard extension of the satisfaction clause for
the $\BoxTy$ modality:
\begin{equation*}
  \begin{array}{l@{\;\Vdash\;} @{\;} l @{\;\text{iff}\;}c@{\;} l}
    \Mod[M],w & \BoxTy{A} & & \text{for all}\ w',v\ \text{s.t.}\ 
        w \Ri w'\ \text{and}\ w' \Rm v,\ \Mod[M],v \Vdash A
  \end{array}
\end{equation*}
it follows that $\Mod[M],w \Vdash A$ if and only if $\Mod[M],w \Vdash
\BoxTy{A}$ for an arbitrary model~$\Mod[M]$ of \LL, making the
connective~$\BoxTy{}$ a logically meaningless addition to \LL.

\emph{Proof\hyp{}relevant semantics}.
Alechina et al.~\cite{AlechinaMPR01} study a connection between
categorical and possible\hyp{}world models of lax logic.
They show that a \LL{}\hyp{}\emph{modal algebra} determines a
possible\hyp{}world model of \LL \cite[Theorem 4]{AlechinaMPR01} via
the Stone representation, and observe that a modal algebra is a
``thin'' categorical model, whose morphisms are given by the
partial\hyp{}order relation of the algebra.
This connection, while illuminating, does not satisfy an important
requirement motivating Mitchell and Moggi's~\cite{MitchellM91} work:
to construct models of lambda calculi by leveraging the
possible\hyp{}world semantics of the corresponding logic.
Our proof\hyp{}relevant possible\hyp{}world semantics satisfies this
requirement and is the key to constructing NbE models.

Kavvos~\cite{Kavvos24a,Kavvos24b} develops proof\hyp{}relevant
possible\hyp{}world semantics (calling it ``Two\hyp{}dimensional
Kripke semantics'') for the modal logic~IK$_{\BlackDiaTy{}}$ of Galois
connections due to Dzik et al~\cite{DzikJK10}, which corresponds to
the minimal Fitch\hyp{}style calculus~\cite{Borghuis94,Clouston18}.
Kavvos adopts a categorical perspective and shows that profunctors
determine an adjunction~$\BlackDiaTy{} \dashv \BoxTy$ on presheaves,
which can be used to model IK$_{\BlackDiaTy{}}$.
Kavvos' profunctor condition is the proof\hyp{}relevant refinement of
Sotirov's~\cite{Sotirov80} bimodule frame condition which states that
${\Ri} ; {\Rm} ; {\Ri} \subseteq {\Rm}$

Proof\hyp{}relevant possible\hyp{}world semantics and its connection
to NbE for modal lambda calculi is a novel consideration in our work.
Valliappan et al~\cite{ValliappanRC22} prove normalization for
Fitch\hyp{}style modal lambda calculi~\cite{Borghuis94,Clouston18},
consisting of the necessity modality~$\BoxTy{}$ and its left
adjoint~$\BlackDiaTy$ using possible\hyp{}world semantics with a
proof\hyp{}\emph{irrelevant} relation~$\Rm$.

\emph{Frame correspondence}.
The study of necessary and sufficient frame conditions for modal
axioms, known as \emph{frame correspondence}, appears to be tricky in
the proof\hyp{}relevant setting.
Plotkin and Stirling~\cite{PlotkinS86} prove a remarkably general
correspondence theorem (Theorem 2.1) that tells us that the
reflexivity of ${\Rm};{\Ri}^{-1}$ corresponds to axiom~\RA and
${\Rm}^2 \subseteq {\Rm};{\Ri}^{-1}$ corresponds to axiom~\JA.
We have not studied frame correspondence in this article, but leave it
as a matter for future work.
The categorical methods of Kavvos~\cite{Kavvos24a} might help here.

%%
%% Bibliography
%%

\bibliography{main}

\appendix

\section{Definitions of strong functors}
\label{app:def}

A \emph{strong} functor~$\FFun : \Cat \to \Cat$ for a cartesian
category~$\Cat$ is an endofunctor on $\Cat$ with a natural
transformation~$\strength_{P,Q} : P \Prod \FFun Q \to \FFun(P \Prod
Q)$ natural in $\Cat$\hyp{}objects~$P$ and~$Q$ such that the following
diagrams stating coherence conditions commute:
\begin{center}
\begin{tikzcd}
  \Term \Prod \FFun{P} \ar{ddr}{\prj_2} \ar{rr}{\strength_{\Term,P}} && \FFun{(\Term\Prod P)} \ar{ddl}[swap]{\FFun{\prj_2}}\\ \\
  & \FFun{P}
\end{tikzcd}
\begin{tikzcd}
  (P\Prod Q) \Prod \FFun{R} \ar{dd}{\assoc_{P,Q,\FFun{R}}} \ar{rr}{\strength_{P\Prod Q, R}} && \FFun{((P\Prod Q)\Prod R)} \ar{dd}{\FFun{\assoc_{P,Q,R}}} \\ \\
  P \Prod (Q\Prod \FFun{R}) \ar{r}{\id_{P}\Prod \strength_{Q,R}} & P\Prod
  \FFun{(Q\Prod R)} \ar{r}{\strength_{P,Q\Prod R}} & \FFun{(P\Prod(Q\Prod R))}
\end{tikzcd}
\end{center}
Observe that the terminal object~$\Term$, the projection
morphism~$\prj_2 : P \Prod Q \to Q$ and the associator morphism
~$\alpha_{P,Q,R} : (P \Prod Q) \Prod R \to P \Prod (Q \Prod R)$ (for
all $\Cat$\hyp{}objects~$P,Q,R$) live in the cartesian
category~$\Cat$.

A \emph{pointed} functor~$\FFun : \Cat \to \Cat$ on a category~$\Cat$
is an endofunctor on $\Cat$ equipped with a natural
transformation~$\point : \Id \natto \FFun$ from the identity functor
$\Id$ on $\Cat$.

A strong and pointed functor~$\FFun$ is said to be \emph{strong
  pointed}, when it satisfies an additional coherence condition
that~$\point$ is a strong natural transformation, meaning that the
following diagram stating a coherence condition commutes:
\begin{center}
\begin{tikzcd}
  & P \Prod Q \ar{ddl}[swap]{\id_{P}\Prod\point_Q} \ar{ddr}{\point_{P\Prod Q}} \\ \\
  P\Prod \FFun{Q} \ar{rr}{\strength_{P,Q}} && \FFun{(P\Prod Q)}
\end{tikzcd}
\end{center}

A \emph{semimonad}~$\FFun : \Cat \to \Cat$, or \emph{joinable}
functor, on a category~$\Cat$ is an endofunctor on $\Cat$ that forms a
semigroup in the sense that it is equipped with a ``multiplication''
natural transformation~$\join : \FFun^2 \natto \FFun$ that is
``associative'' as
$\join_{\PshObj{P}} \after \join_{\PshObj{\FFun \PshObj{P}}} =
\join_{\PshObj{P}} \after \FFun (\join_{\PshObj{P}}) : \FFun^3
\PshObj{P} \to \FFun \PshObj{P}$.

A strong functor~$\FFun$ that is also a semimonad is a \emph{strong
semimonad} when $\join$ is a strong natural transformation, meaning
that the following coherence condition diagram commutes:

\begin{center}
\begin{tikzcd}
  \PshObj{P}\Prod \FFun{\FFun{\PshObj{Q}}} \ar{dd}{\id_{\PshObj{P}}\Prod\mu_\PshObj{Q}} \ar{r}{\strength_{\PshObj{P},\FFun{\PshObj{Q}}}} & \FFun{(\PshObj{P}\Prod \FFun{\PshObj{Q}})} \ar{r}{\FMap{\strength_{\PshObj{P},\PshObj{Q}}}} & \FFun{\FFun{(\PshObj{P}\Prod \PshObj{Q})}} \ar{dd}{\mu_{\PshObj{P}\Prod \PshObj{Q}}} \\ \\
  \PshObj{P}\Prod \FFun{\PshObj{Q}}
  \ar{rr}{\strength_{\PshObj{P},\PshObj{Q}}} && \FFun{(\PshObj{P}\Prod
    \PshObj{Q})}
\end{tikzcd}
\end{center}

A \emph{monad}~$\FFun : \Cat \to \Cat$ on a category~$\Cat$ is a
semimonad that is pointed, such that the natural
transformation~$\point : \Id \natto \FFun$ is the left and right unit
of multiplication~$\join : \FFun^2 \natto \FFun$ in the sense that
$\join_{\PshObj{P}} \after \FMap{\point_{\PshObj{P}}} =
\id_{\FFun{\PshObj{P}}}$ and
$\join_{\PshObj{P}} \after \point_{\FFun{\PshObj{P}}} =
\id_{\FFun{\PshObj{P}}}$ for some $\Cat$\hyp{}object~$P$.

A strong functor~$\FFun$ that is also a monad is a \emph{strong monad}
when the natural transformations~$\point$ and~$\join$ of the monad are
both strong natural transformations, making $\FFun$ both a strong
pointed functor and a strong semimonad.

\section{Auxiliary definitions}

\begin{definition}[Context Inclusion]
The relation~$\Wk{}{}$ is defined inductively on contexts:
\begin{mathpar} 
  \inferrule[]{%
    \  
  }{%
      \baseWk{} : \Wk{\EmptyCtx}{\EmptyCtx}
  }\
  
  \inferrule[]{%
    i : \Wk{\Gamma}{\Gamma'}\\ (x\ \text{not in}\ \Gamma')
  }{%
    \dropWk_A{i} : \Wk{\Gamma}{\ExtCtx{\Gamma'}{x : A}}
  }
  
  \inferrule[]{%
    i : \Wk{\Gamma}{\Gamma'}\\ (x\ \text{not in}\ \Gamma')
  }{%
    \keepVarWk_A{i} :\Wk{\ExtCtx{\Gamma}{x : A}}{\ExtCtx{\Gamma'}{x : A}}
  }
\end{mathpar}
\end{definition}

\begin{itemize}
\item The relation $\Wk{}{}$ is reflexive and transitive, as witnessed by functions:
\begin{align*}
  & \begin{array}{l @{\;}c@{\;} l}
    \multicolumn{3}{l}{\idWk :  \forall \Gamma.\, \Wk{\Gamma}{\Gamma}} \\
    \idWk_{\EmptyCtx}         & = & \baseWk{} \\
    \idWk_{\ExtCtx{\Gamma}{x : A}} & = & \keepVarWk_A{\idWk_{\Gamma}}
  \end{array}
  & \begin{array}{l l l @{\;}c@{\;} l}
    \multicolumn{5}{l}{\transWk : \forall \Gamma,\Gamma',\Gamma''.\, \Wk{\Gamma}{\Gamma'} \to \Wk{\Gamma'}{\Gamma''} \to \Wk{\Gamma}{\Gamma''}} \\
    \transWk & i               & \baseWk{}        & = & i \\
    \transWk & i               & (\dropWk{i'})    & = & \dropWk{(\transWk{i, i'})} \\
    \transWk & (\dropWk{i})    & (\keepVarWk{i'}) & = & \dropWk{(\transWk{i,i'})} \\
    \transWk & (\keepVarWk{i}) & (\keepVarWk{i'}) & = & \keepVarWk{(\transWk{i,i'})}
  \end{array}
\end{align*}
\item The function $\factor$ for the NbE model of \SLC is defined as:
\begin{equation*}
  \begin{array}{l l @{\;}c@{\;} l}
    \multicolumn{4}{l}{\factor : \forall \Gamma,\Gamma',\Delta.\, \Wk{\Gamma}{\Gamma'} \to \RSL{\Gamma}{\Delta} \to \exists
  \Delta'.\, (\RSL{\Gamma'}{\Delta'} \times \Wk{\Delta}{\Delta'})} \\
    \factor & i & (\varExtOnce{n}) = & \pair{\varExtOnce{(\wkTm{n})}}{\keepVarWk{i}}
   \end{array}
 \end{equation*}
\item The function $\incl$ for the NbE model of \SLC is defined as:
\begin{equation*}
  \begin{array}{l @{\;}c@{\;} l}
    \multicolumn{3}{l}{\incl : \forall \Gamma, \Delta.\, \RSL{\Gamma}{\Delta} \to \Wk{\Gamma}{\Delta}} \\
    \incl & (\varExtOnce{(n : \Ne{\Gamma}{\DiaTy{A}})}) = & \dropWk_A{\idWk_\Gamma}
   \end{array}
 \end{equation*}

\end{itemize}

\end{document}